\documentclass[10.5pt,prl,aps,twocolumn]{revtex4-2}

\usepackage{graphicx,epsfig,amsmath,amssymb,verbatim,color}
\usepackage{titlesec}
\usepackage{dsfont}
\usepackage{float}
\usepackage{tikz}
\usepackage{pgfplots}
\usepackage{xcolor}
\definecolor{darkred}{rgb}{0.8,0.1,0.1}
\usepackage{soul}
\usepackage{amsmath,mathtools}
\usepackage{enumerate}
\usepackage[T1]{fontenc}
\usepackage[utf8]{inputenc}
\usepackage{pifont}
\usepackage{hyperref}
\usepackage{appendix}
\hypersetup{colorlinks=true,citecolor=blue,linkcolor=blue,filecolor=blue,urlcolor=blue,breaklinks=true}
\usepackage[utf8]{inputenc}  
\usepackage{fontenc}         
\usepackage{colortbl}
\usepackage{pifont}
\definecolor{Gray}{gray}{0.92}
\definecolor{Gray2}{gray}{0.75}
\definecolor{maroon}{cmyk}{0,0.87,0.68,0.32}
\usepackage{booktabs}
\usepackage{makecell}
\usepackage{diagbox}
\usepackage{multirow}

\newtheorem{proposition}{Proposition}
\newtheorem{lemma}[proposition]{Lemma}

\newtheorem{theorem}[proposition]{Theorem}
\newtheorem{remark}{Remark}
\newtheorem{example}[proposition]{Example}
\newtheorem{corollary}[proposition]{Corollary}

\newenvironment{proof}{\noindent \textit{{Proof.}~}}{\hfill $\square$}

\def\squareforqed{\hbox{\rlap{$\sqcap$}$\sqcup$}}
\def\qed{\ifmmode\squareforqed\else{\unskip\nobreak\hfil
		\penalty50\hskip1em\null\nobreak\hfil\squareforqed
		\parfillskip=0pt\finalhyphendemerits=0\endgraf}\fi}
\def\endenv{\ifmmode\;\else{\unskip\nobreak\hfil
		\penalty50\hskip1em\null\nobreak\hfil\;
		\parfillskip=0pt\finalhyphendemerits=0\endgraf}\fi}

\newcommand{\bra}[1]{\langle#1|}
\newcommand{\ket}[1]{|#1\rangle}

\def\Dbar{\leavevmode\lower.6ex\hbox to 0pt
	{\hskip-.23ex\accent"16\hss}D}

\begin{document}
	\title{Conditional Error Exponents of Probabilistic Quantum Resource Distillation}
	
	\author{Xian Shi}\email[]
	{shixian01@gmail.com}
	\affiliation{College of Information Science and Technology,
		Beijing University of Chemical Technology, Beijing 100029, China}

	\date{\today}
	\begin{abstract}
Quantum resource distillation converts imperfect resource states into highly resourceful targets and is a central task in quantum information processing. When distillation is probabilistic, unfavorable outcomes may be discarded and the quality of the output is evaluated conditioned on success, leading to a notion of reliability that differs from the deterministic setting. Here we study the conditional error exponents of probabilistic quantum resource distillation under approximately resource-nongenerating instruments and derive its general connection to the postselected composite quantum hypothesis testing against the corresponding set of free states. This reduction leads naturally to scale-invariant projective distinguishability measures and provides finite-blocklength bounds for a broad class of resource theories. We apply the framework to entanglement, coherence, and magic. For entanglement and magic, we obtain exact zero-rate asymptotic characterizations of the conditional error exponent in terms of regularized projective distinguishability, while for coherence we derive general finite-blocklength and asymptotic bounds. For representative entangled and magic states, the resulting quantities can be evaluated analytically and reveal regimes in which postselection yields a strictly larger conditional error exponent than deterministic distillation. Furthermore, we also show that the conditional error exponents for probabilistic entanglement distillation under a strictly positive success probability constraint are smaller than those with no constraints on the success probability.

	\end{abstract}

	\maketitle

	\section{Introduction}

	Quantum information science exploits quantum intrinsic features to guarantee information-processing capabilities that are inaccessible using classical resources only \cite{Huang2026Vast}. Such quantum advantages arise across a broad range of tasks, including computation \cite{LaRose2026Brief,Babbush2026Grand}, communication \cite{Amiri2024Quantum}, learning \cite{Anshu2024Survey,Yamasaki2026Advantage,Zhao2025Entanglement}, and sensing \cite{Bass2024Quantum}, and constitute a central motivation for the development of quantum technologies. Understanding the operational value of these advantages naturally motivates the study of quantum resources.
	
	 Quantum resources, including entanglement \cite{horodecki2009quantum}, coherence \cite{streltsov2017colloquium}, and magic \cite{Veitch_2014, howard2017application}, underpin many of the advantages offered by quantum information processing. In realistic implementations, however, these resources are inevitably degraded by noise, making their direct use inefficient or even impossible. A fundamental task in quantum resource theories is therefore resource distillation, whereby many copies of a noisy resource state are converted, using only free operations, into a smaller number of high-quality standard resource states \cite{Bennett1996Purification,winter2016,BravyiKitaev2005,Gour2009,CampbellAnwarBrowne2012,chitambar2019quantum}. Beyond specific resource theories, resource distillation have been developed in a common framework of general quantum resource theories \cite{Liu2019one,Kuroiwa2020General,regula2020benchmarking,PhysRevLett.125.060405}. 
	A central objective of resource distillation is to determine the optimal asymptotic conversion rate. In \cite{PhysRevLett.115.070503,hayashi2025general,hayashi2025generalized,lami2025a}, the authors showed the asymptotic conversion rates under the asymptotically resource-nongenerating operations are intimately connected with regularized relative entropy measure of resources by building the connections between the resource distillation and quantum hypothesis testing. Except for the achievable rate of the resource distillation, it is also important to understand how rapidly the conversion error decreases with the number of input copies by studying error exponent of the deterministic distillation, it has been solved partly in  \cite{rippchen2025fundamental,lami2026,Watanabe2026rela}, where the authors showed the relation between the error exponents of resource distillation under resource nongenerating operations and the Sanov exponent of quantum resource testing.

Probabilistic quantum resource transformations have been studied extensively in general resource-theoretic settings. In particular, projective resource monotones have been shown to provide strong constraints on probabilistic state transformations and, in broad classes of resource theories, exact criteria for one-shot convertibility and distillation \cite{regula2022probabilistic,regula2022tight,regula2023overcoming}. Besides, postselected quantum hypothesis testing has provided an operational interpretation of projective distinguishability: allowing an inconclusive outcome and conditioning the discrimination error on a conclusive event leads, in both simple and composite hypothesis-testing problems, to Hilbert-projective-type asymptotic quantities \cite{regula2023postselected}.
The present work addresses a different question within this broader picture. Rather than optimizing the asymptotic conversion rate or the success probability of a probabilistic transformation, we study its conditional reliability: for a prescribed distillation target, how rapidly can the error of the postselected output decrease with the input blocklength? We formulate this question for approximately resource-nongenerating instruments and relate the resulting finite-blocklength conditional error to postselected composite hypothesis testing against the free-state set. This yields a general one-shot reduction and, in specific resource theories, explicit characterizations of the corresponding conditional error exponents.
Our previous work \cite{11654091} considered this reliability problem specifically for entanglement distillation under approximately nonentangling instruments, where postselected testing against separable states gives analytical characterizations of the entanglement-distillation exponent. The present work differs in both scope and operational formulation. First, we formulate the conditional-reliability problem for a broader class of quantum resource theories and quantify approximate resource nongeneration through generalized robustness. Second, we extend the analysis beyond entanglement to coherence and magic. For entanglement and magic, the resulting postselected-testing description yields exact finite-blocklength and zero-rate asymptotic characterizations, whereas for coherence it provides general bounds. Third, we obtain explicit evaluations for representative entangled and magic states and identify regimes exhibiting a strict separation between probabilistic and deterministic conditional reliability.
Thus, the contribution of this work is not the introduction of projective quantities or postselected hypothesis testing themselves, but their systematic use in characterizing the conditional error exponents of approximately resource-nongenerating distillation across distinct quantum resource theories. Fourth, we obtain the bound of conditional error exponents of probabilistic entanglement distillation when the success probability is required greater than 0, furthermore, by comparing with the scenario when the success probability is unconstrained for the class of Werner states, we show the differences of the scenario when the success probability is positive.

	This article is organized as follows. In Sec. \ref{pk}, we introduce the resource theories with requisite properties considered here, we also recall the quantifiers and operations of quantum resource theories needed here. In Sec. \ref{mee}, we first build the relation between the conditional error exponent of one-shot probabilistic resource distillation under the approximately nongenerating instruments and the task of quantum postselected composite hypothesis testing. In Sec. \ref{m}, 
	we consider the conditional error exponents of entanglement, coherence and magic distillation in terms of $n$-blocklength and asymptotic scenarios with zero rate under the resource nongenerating instruments. And in the resource theory of entanglement and magic, we obtain the analytical formulae of conditional error exponents for a class of states under the approximately resource nongenerating instruments.
Furthermore, we reveal the advantages of postselection under quantum resource distillation scenario by comparing with the error exponents of examples in the resource of entanglement and magic without postselection, we also obtain the bounds of the conditional error exponents of probabilistic entanglement distillation when restraining the probability greater than 0. In Sec. \ref{con}, we end with a conclusion.

	\section{Preliminary Knowledge}\label{pk}

Assume $\mathcal{H}$ is a Hilbert space with $2\le dim \mathcal{H}<\infty$. A generic resource theory consists of free state set $\mathbf{F}$ and free operation set $\mathbf{O}.$ The minimal requirement for free operations is that 
\begin{align}
	\Lambda\in \mathbf{O}\Longrightarrow \Lambda(\sigma)\in \mathbf{F}, \forall \sigma\in \mathbf{F}.\label{p1}
\end{align}
Next we denote the set consisting of all the operations satisfying (\ref{p1}) as the resource nongenerating operations and write it as $\mathbf{O}_{max}$ \cite{chitambar2019quantum,shi2026}. Apparently, each valid set of free operations is a subset of $\mathbf{O}_{max}.$ 

A key tool to study generic resource theories is the measure which quantifies the resource of a state. It is a family of functions which map the states to $\mathbb{R}^{+}\cup\{0\}.$ The necessities for the quantifiers is that the image of free states under $f$ equal to 0 and that they donot increase under the application of free operations.
Here we introduce the following resource measures. The first is the generalized robustness of $\langle\mathbf{F},\mathbf{O}_{max}\rangle$ \cite{PhysRevX.9.031053}, which is defined as 
\begin{align*}
	R_{\mathbf{F},g}(\rho)= \inf\{s|\frac{\rho+s\tau}{1+s}\in \mathbf{F},\tau\in \mathcal{D}_{\mathcal{H}}\},
\end{align*}
where $\tau$ takes over all the elements in $\mathcal{D}_{\mathcal{H}}$. The other measure
for $\rho$ is on the distance between $\rho$ and the set of free states based on fidelity,
\begin{align*}
	D_{\mathbf{F},F}(\rho)=& 1-\max_{\sigma\in \mathbf{F}} F(\rho,\sigma)\\
	=&1-\max_{\sigma\in \mathbf{F}}\mathrm{ tr}||\sqrt{\rho}\sqrt{\sigma}||_1.
\end{align*}
Here the maximum takes over all the free states $\sigma\in\mathbf{F}$. Next we list the properties that the statical resource theories $\langle\mathbf{F},\mathbf{O}_{max}\rangle$ considered here. 
\begin{itemize}
	\item[(1).] $\mathbf{F}$ and $\mathbf{O}_{max}$ are convex.
	\item[(2).] There exists a pure reference state $\ket{\psi}$.
\item[(3).]	For any system with dimension $d$, the reference state $\ket{\psi}$ satisfies the following,
	\begin{align}
		\mathrm{ tr}\ket{\psi}\bra{\psi}\sigma\le c_d\le\frac{1}{2},\hspace{3mm}\forall\sigma\in \mathbf{F},\label{p3}
	\end{align}
	where $c_d$ is a nonincreasing function on $d$.
	\item[(4).] $\frac{I}{d}\in \mathbf{F}.$
\end{itemize}

\begin{remark}
	This inequality $(\ref{p3})$ is satisfied by common resource theories.
	\begin{itemize}
		\item For the resource theory of bipartite entanglement, let $\ket{\psi}=\frac{1}{\sqrt{d}}\sum_{i=1}^d\ket{ii}$, then $$\mathrm{tr}\Psi \sigma\le \frac{1}{d},\hspace{4mm}\forall \sigma\in Sep.$$
		\item For the theory of coherence, let $\ket{\psi}=\frac{1}{\sqrt{d}}\sum_{i=1}^d\ket{i},$ then $$\mathrm{tr}\Psi \sigma=\frac{1}{d}\hspace{4mm} \forall \sigma\in \mathcal{I},$$ here $\mathcal{I}$ is the set of incoherent states in terms of the given referenced basis.
		\item For the resource theory of asymmetry, when the dimension of the system is $d+1$ and the symmetry group is $G=U(1)$, let $\ket{\psi}=\frac{1}{\sqrt{d+1}}\sum^{d}_{i=0}\ket{i}$, then $$\mathrm{tr}\ket{\psi}\bra{\psi}\sigma\le \frac{1}{d+1},$$ here $\sigma$ is any symmetric state.
	\end{itemize}  
\end{remark}

Assume $\mathcal{E}=\{\mathcal{E}_i\}$ is a set of completely positive and trace nonincreasing maps such that $\sum_i\mathcal{E}_i$ is a channel, then $\mathcal{E}$ is $\delta$-approximately resource nongenerating instruments if and only if 
\begin{align*}
	 R_{\mathbf{F},g}(\frac{\mathcal{E}_i(\sigma)}{\mathrm{tr}\mathcal{E}_i(\sigma)})\le \delta,\hspace{4mm}\forall\mathcal{E}_i\in \mathcal{E},\sigma\in \mathbf{F}.
\end{align*}
And we denote the set of $\delta$-approximately resource nongenerating instruments as
\begin{align*}
	\mathbb{O}_{\mathcal{NG}}^{\delta}=\{\mathcal{E}=\{\mathcal{E}_i\}|R_{\mathbf{F},g}(\frac{\mathcal{E}_i(\sigma)}{\mathrm{tr}\mathcal{E}_i(\sigma)})\le \delta, \forall\sigma\in \mathbf{F},\forall\mathcal{E}_i\in \mathcal{E}\},
\end{align*}
where $R_{\mathbf{F},g}(\cdot)$ is the generalized robustness of $\langle\mathbf{F},\mathbf{O}_{max}\rangle$. Besides, when $\mathcal{L}$ is a channel, if
\begin{align*}
	R_{\mathbf{F},g}({\mathcal{L}(\sigma)})\le \delta,\hspace{4mm}\forall\sigma\in \mathbf{F}.
\end{align*}
we say $\mathcal{L}$ is a $\delta$-approximately resource nongenerating channel. 
The set of $\delta$-approximately resource nongenerating channels is denoted as
\begin{align*}
	{\mathcal{NG}}_{\delta}=\{\mathcal{L}|R_{\mathbf{F},g}(\mathcal{L}(\sigma))\le \delta, \forall\sigma\in \mathbf{F}\}.
\end{align*}

Next the error exponents for probabilistic transformations from $\rho$ into $\ket{\psi}$ under $\mathbb{O}_{\mathcal{NG}}^{\delta}$ in terms of one-shot scenario is 
\begin{align*}
	E_{d,err,p}(\rho)=&\sup -\log\epsilon\\
	\textit{s. t.}\hspace{4mm}&	F(\frac{\mathcal{E}_i(\rho)}{\mathrm{tr}(\mathcal{E}_i(\rho))},\ket{\psi}\bra{\psi})\ge 1-\epsilon\hspace{2mm}\exists \mathcal{E}_i\in\mathcal{E},\\
	&\mathcal{E}=\{\mathcal{E}_i\}\in \mathbb{O}_{\mathcal{NG}}^{{\delta}}.
\end{align*}
where $\ket{\psi}$ is the reference state, and the supermum takes over all the instruments $\mathcal{E}\in \mathbb{O}_{\mathcal{NG}}^{\delta}$.

\section{Bound of Error Exponents of the Transformation into Reference States}\label{mee}
 Assume $\ket{\psi}$ is the reference state, and there exists a finite or a compact Lie group $G$ with a unitary repreesntation $\{U_g\}_{g\in G}$ that satisfies $U_g(\cdot)U_g^{\dagger}\in \mathbf{O}_{max}$ and $U_g\ket{\psi}=e^{i\phi_g}\ket{\psi}$ $\forall g\in G$ with some eigenvalue $e^{i\phi_g}$. Accordingly, the decomposition of the representation $\{U_g\}$ is associated with the decomposition of the Hilbert space:
\begin{align*}
	\mathcal{H}=\bigoplus_{\mu} \mathcal{H}_{\mu}^{(1)}\otimes\mathcal{H}_{\mu}^{(2)},
\end{align*}
where $\mu$ is the number of the irreducible representation, $\mathcal{H}_{\mu}^{(1)}$ is the subspace on which each irreducible represenation acts nontrivially, and $\mathcal{H}_{\mu}^{(2)}$ denotes the multiplicity subspace. Furthermore, when $\dim\mathcal{H}_{\mu}^{(2)}=1,$ $\forall \mu,$ that is, 
\begin{align}
	\mathcal{H}=\bigoplus_{\mu} \mathcal{H}_{\mu}, \label{hs}
\end{align}

Here we present the conditional error exponents $E_{d,err,p}(\rho)$ from $\rho$ to $\psi$ probabilistically under a class of resource theories with the properties above.

\begin{theorem}\label{th1}
	Assume $\rho$ is a state on $\mathcal{H}$. Let $\ket{\psi}$ be the reference state and $1>\delta> 0$, if there exists a finite or a compact Lie group $G$ with a unitary representation $\{U_g\}_{g\in G}$ such that $\mathcal{U}_g(\cdot)=U_g(\cdot)U_g^{\dagger}\in \mathbf{O}_{max}$ and $\ket{\psi}$ is the eigenvector of any $U_g$, the associated decomposition of $\mathcal{H}$ satisifies $(\ref{hs})$, then a lower bound of the conditional error exponent for the probabilistic transformations from $\rho$ into $\ket{\psi}$ under $\mathbb{O}_{\mathcal{NG}}^{\delta}$ is
\begin{align*}
	E_{d,err,p}(\rho)\ge \hat{\beta}_{\frac{\delta+1}{d},\mathbf{F}}(\rho).
\end{align*}
Here $d$ is the dimension of $\mathcal{H}.$
\end{theorem}

The proof of Theorem \ref{th1} is placed in the Appendix \ref{ath1}.

	\section{Applications}\label{m}
	
	In this section, we consider the conditional error exponents of the resource distillation in terms of $n$-blocklength and asymptotic scenarios for the resource theories  of entanglement, coherence and magic with postselected quantum composite hypothesis testing. And we present analytical expressions of conditional error exponents of probabilistic distillation for a class of quantum states under the resource of entanglement and magic, which shows the advantages of postselection of quantum resource distillation. Furthermore, we also show that the conditional error exponents for probabilistic entanglement distillation under a strictly positive success probability constraint are smaller than those with no constraints on the success probability.
	
	\subsection{The resource theory of entanglement}
	Assume $\mathcal{H}_{AB}$ is a bipartite system with finite dimensions. A state $\rho$ is separable if it can be written as $\rho=\sum_ip_i\rho_i^A\otimes\rho_i^B$, otherwise, it is entangled. Here we denote the set of separable states of $\mathcal{H}_{AB}$ as $Sep_{A:B}$. Besides, we denote $\overline{Sep}_{A:B}$ as the set of separable substates on $\mathcal{H}_{AB}$. When $dim(\mathcal{H}_A)=dim(\mathcal{H}_B)=d$, the maximally entangled state is $\ket{\psi}_d=\frac{1}{\sqrt{d}}\sum_{i=1}^d\ket{ii}$. One of the most common used  entanglement quantifiers is generalized robustness of entanglement, $R_{G}(\cdot)$ \cite{vidal1999robustness,steiner2003generalized}. Assume $\rho_{AB}$ is a bipartite state, its generalized robustness of entanglement is
	\begin{align*}
		R_G(\rho)=\min\{s|\frac{\rho+s\sigma}{1+s}\in Sep, \sigma\in \boldsymbol{D}({\mathcal{H}_{AB}})\}.
	\end{align*}

	Let $\mathcal{E}_i$ be completely positive and trace non-increasing, if  
	\begin{align*}
		\rho\in Sep_{A:B}\Longrightarrow	\frac{\mathcal{E}_i(\rho)}{\mathrm{tr}\mathcal{E}_i(\rho)}\in Sep_{A:B}, \forall i,
	\end{align*}
	then $\mathcal{E}_i$ is nonentangling($\mathcal{NE}$).  If $\{\mathcal{E}_i\}$ is a set of $\mathcal{NE}$ subchannels and $\mathrm{tr}\sum_i\mathcal{E}_i(\cdot)=\mathrm{ tr}(\cdot),$ then $\{\mathcal{E}_i\}$ is a $\mathcal{NE}$ instrument. Here we denote the set of all such $\mathcal{NE}$ instruments as $\mathbb{O}_{\mathcal{NE}}.$ Next for a subchannel $\Lambda(\cdot)$, we can look at the Hisenberg picture, where $\Lambda^{\dagger}(\cdot)$ satisfies the following property, $\mathrm{ tr}X\Lambda(Y)=\mathrm{ tr}\Lambda^{\dagger}(X)Y$, for any $X$ and $Y.$ If $\Lambda$ satisfies the following property,
	\begin{align*}
		\Lambda(\rho)\in cone(Sep_{A:B}),\hspace{3mm}\forall\rho\in Sep_{A:B},\\
		\Lambda^{\dagger}(\rho)\in cone(Sep_{A:B}),\hspace{3mm} \forall\rho\in Sep_{A:B},
	\end{align*}
	then we say $\Lambda$ is dually nonentangling($\mathcal{DNE}$). If $\{\Lambda_i\}$ is a set of subchannels with each $\Lambda_i$ $\mathcal{DNE}$ and $\mathrm{ tr}\sum_i\Lambda_i(\cdot)=\mathrm{ tr}(\cdot)$, $\{\Lambda_i\}$ is a $\mathcal{DNE}$ instrument. The set of all such $\mathcal{DNE}$ instruments are denoted as $\mathbb{O}_{DNE}$. Following the work of Brandao and Plenio \cite{brandao2010reversible}, we can also define the set of asymptotically $\mathcal{NE}$ and $\mathcal{DNE}$ subchannels, respectively. Assume $\{\mathcal{E}_i\}$ is a set of subchannels such that $\sum_i\mathcal{E}_i$ is a channel, then the $\delta$-approximately non-entangling quantum instruments, $\mathcal{E}=\{\mathcal{E}_i|\mathcal{E}_i\in \boldsymbol{CP}, \sum_i\mathcal{E}_i\in\boldsymbol{C}\}$ is defined as
	\begin{align*}
		\mathbb{O}_{\mathcal{NE}}^{\delta}=\{\mathcal{E}|R_G(\frac{\mathcal{E}_i(\sigma)}{\mathrm{tr}\mathcal{E}_i(\sigma)})\le \delta, \forall \mathcal{E}_i\in \mathcal{E},\sigma\in Sep\}.
	\end{align*}
	where $R_G(\sigma)$ is the generalized robustness of entanglement. Analogously, the $\delta$-approximately dually non-entangling quantum instruments is defined as
	\begin{align*}
		\mathbb{O}_{\mathcal{DNE}}^{\delta}=\{\mathcal{E}|\mathcal{E}_i^{\dagger}(Sep)\subset cone(Sep),\forall\mathcal{E}_i\in \mathcal{E}\}\cap\mathbb{O}_{\mathcal{NE}}^{\delta}.
	\end{align*}

	Entanglement transformation is a fundamental task in quantum entanglement theory. When the dimension of target maximally entangled state is fixed, the probabilitic error exponent for $\rho$ in terms of $n$-blocklength scenario under  $\mathbb{O}_{\mathcal{F}}^{\delta}$ instruments $\{\mathcal{E}_i\}$ is defined as follows,
	\begin{align*}
		E_{d,err,p}^{(m),\mathbb{O}_{\mathcal{F}}^{\delta}}(\rho^{\otimes n}_{AB})=&\sup-\frac{1}{n}\log\epsilon_n\\
		\textit{s. t.}\hspace{4mm}&	F(\frac{\mathcal{E}_i(\rho_{AB}^{\otimes n})}{\mathrm{tr}(\mathcal{E}_i(\rho_{AB}^{\otimes n}))},\Psi_m)\ge 1-\epsilon_n\hspace{2mm}\exists \mathcal{E}_i\in\mathcal{E},\\
		&\mathcal{E}\in \mathbb{O}_{\mathcal{F}}^{\delta},\hspace{3mm}\mathcal{F}=\{\mathcal{NE},\mathcal{DNE}\}
	\end{align*}
	where $\Psi_m=\frac{1}{m}\sum\ket{ii}\bra{jj}$, $F(\rho,\sigma)=||\sqrt{\rho}\sqrt{\sigma}||^2_1,$ and the supermum takes over all instruments $\{\mathcal{E}_i\}\in \mathbb{O}_{\mathcal{F}}^{\delta}.$ Moreover,
	the zero-rate conditional error exponent of probabilistic entanglement distillation under $\mathbb{O}_{\mathcal{F}}^{\delta}$ for $\rho$ is 
	\begin{align*}
		E^{\infty,\mathbb{O}^{\delta}_{\mathcal{F}}}_{d,err,p}(\rho_{AB}):=\liminf\limits_{n\rightarrow\infty} \sup E_{d,err,p}^{(m_n),\mathbb{O}_{\mathcal{F}}^{\delta}}(\rho^{\otimes n}_{AB}),
	\end{align*}
	where $\{m_n\}_n$ takes over all sequences with $\lim\limits_{n\rightarrow\infty}\frac{\log m_n}{n}=0.$
	
	Next we will present the analytical formulae for the zero-rate conditional error exponents of probabilistic entanglement distillation under $\mathbb{O}^{\delta}_{\mathcal{NE}}$ and $\mathbb{O}^{\delta}_{\mathcal{DNE}}$ instruments for a bipartite state.
\begin{theorem}\label{t1}
Assume $\rho_{AB}^{\otimes n}$ is a bipartite state on $\mathcal{H}_{AB}^{\otimes n}$. Let $m\ge2\in \mathbb{N}$, $\delta\in (0,1)$, then the conditional error exponent of the entanglement distillation under $\mathbb{O}_{\mathcal{NE}}^{\delta}$  is 
\begin{align*}
	E_{d,err,p}^{(m),\mathbb{O}_{\mathcal{NE}}^{\delta}}(\rho^{\otimes n}_{AB})=\frac{1}{n}\hat{\beta}_{\frac{\delta+1}{m},Sep}(\rho_{AB}^{\otimes n}).
\end{align*}
Besides, the zero-rate conditional error exponent of probabilistic entanglement distillation under $\mathbb{O}_{\mathcal{NE}}^{\delta}$ for $\rho_{AB}$ is 
\begin{align*}
	E_{d,err,p}^{\infty,\mathbb{O}_{\mathcal{NE}}^{\delta}}(\rho_{AB})=\lim\limits_{n\rightarrow\infty}E_{d,err,p}^{(m_n),\mathbb{O}_{\mathcal{NE}}^{\delta}}(\rho^{\otimes n}_{AB})=\hat{D}_{\Omega,Sep}^{reg }(\rho).
\end{align*}
\end{theorem}

The proof of Theorem \ref{t1} is placed in the Appendix \ref{aa0}.

\begin{example}\label{e1}
	Assume $\mathcal{H}_{AB}$ is a bipartite system with $dim(\mathcal{H}_A)=dim(\mathcal{H}_B)=d$, and $\rho_{AB}$ is the Werner state,
	\begin{align*}
		\rho_{p}=p\cdot\frac{2P_s}{d(d+1)}+(1-p)\cdot\frac{2P_{as}}{d(d-1)},
	\end{align*}
	here $P_s=\frac{I+F}{2}$, $P_{as}=\frac{I-F}{2}$, $F$ is the swap operator, $F=\sum_{ij}\ket{ij}\bra{ji}$. Then for each $n\in\mathbb{N}$,
	\begin{equation*}
	E^{\infty,\mathbb{O}_{\mathcal{NE}}^{\delta}}_{d,err,p}(\rho_p)=	\frac{1}{n}D_{\Omega,Sep}(\rho_p^{\otimes n})= \begin{cases} 
		\log\frac{1-p}{p}\hspace{3mm}p<\frac{1}{2} \\
			0\hspace{12mm}p\ge\frac{1}{2}
		\end{cases}  .
	\end{equation*}
\end{example}
The proof of Example \ref{e1} is proved in \cite{11654091}. For the sake of completeness of the paper, we place the proofs in the Appendix \ref{aa1}. 

\begin{figure}[h] 
	\centering
	\includegraphics[width=0.5\textwidth]{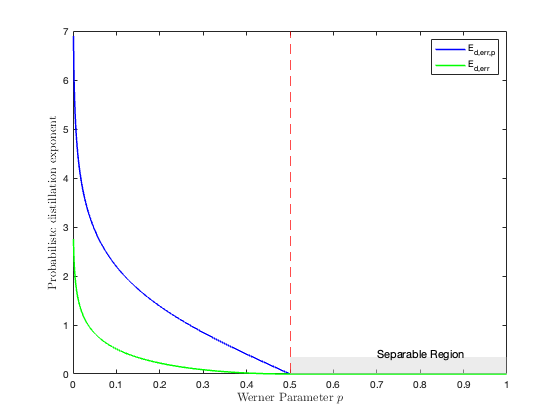} 
	\caption{The asymptotic error exponent of probabilistic entanglement distillation of $\rho_p$ under NE.}
	\label{fig1}
\end{figure}

 	Based on the results in \cite{lami2024}, the conditional error exponent of distillable entanglement under non-entangling operations without postselection equals to
 	\begin{align}
 		E^{\infty,\mathbb{O}_{\mathcal{NE}}^{\delta}}_{d,err}(\rho_{AB})=\min_{\sigma\in Sep } D(\sigma_{AB}||\rho_{AB})=-\frac{\log[4p(1-p)]}{2},\label{derr}
 	\end{align}
 	the proof of (\ref{derr}) is placed in Appendix \ref{aa1}. As when $p\in (0,\frac{1}{2})$, $$\log\frac{1-p}{p}>-\frac{1}{2}\log[4p(1-p)],$$ this demonstrates that postselection can strictly improve the conditional error exponent of entanglement distillation over its deterministic counterpart, albeit at the expense of a reduced probability of success, which is shown in Fig \ref{fig1}.

 \begin{example}\label{e2}
 	Assume $\rho$ is a two-qubit state with 
 	\begin{align*}
 		\rho=p_1\Psi^{+}+p_2\Psi^{-}+p_3\Phi^{+}+p_4\Phi^{-},\\
 		\Psi^{+}=\frac{1}{2}(\ket{00}+\ket{11})(\bra{00}+\bra{11}),\\
 		\Psi^{-}=\frac{1}{2}(\ket{00}-\ket{11})(\bra{00}-\bra{11}),\\
 		\Phi^{+}=\frac{1}{2}(\ket{01}+\ket{10})(\bra{01}+\bra{10}),\\
 		\Phi^{-}=\frac{1}{2}(\ket{01}-\ket{10})(\bra{01}-\bra{10}),
 	\end{align*}
 	where $p_i\ge 0$ and $\sum_i p_i=1.$ Then for each $n\in \mathbb{N}$,
 \begin{align*}
 	E^{\infty,\mathbb{O}_{\mathcal{NE}}^{\delta}}_{d,err,p}(\rho)=	  \begin{cases} 
 		\log\frac{p_{max}}{1-p_{max}}\hspace{3mm}p_{max}>\frac{1}{2} \\
 		0\hspace{18mm}p_{max}\le\frac{1}{2}
 	\end{cases}.
 \end{align*}
 Here $p_{max}=\max\{p_1,p_2,p_3,p_4\}$.  The proof of Example \ref{e2} is placed in the Appendix \ref{aa1}. Based on \cite{lami2024}, when $p_{max}> \frac{1}{2}$,
 \begin{align*}
 	E^{\infty,\mathbb{O}_{\mathcal{NE}}^{\delta}}_{d,err}(\rho)=\min_{\sigma\in Sep_{A:B}} D(\sigma||\rho)=-\frac{\log[4p_{max}(1-p_{max})]}{2},
 \end{align*}
when $p_{max}\in (\frac{1}{2},1)$, $$\log\frac{p_{max}}{1-p_{max}}\ge -\frac{\log[4p_{max}(1-p_{max})]}{2},$$ this example also demonstrates that postselection can strictly improve the conditional error exponent of entanglement distillation over its deterministic counterpart.
 \end{example}
 
 The present results impose no lower bound on the probability of the selected branch. This naturally raises the problem, whether the conditional error advantage persists when the success probability is required to remain bounded away from 0. 
 
 \begin{theorem}\label{thepb0}
 	Assume $\rho_{AB}$ is a bipartite state, and $\delta\in (0,1)$, then 
 	\begin{align*}
 	&E_{d,\mathrm{err},p}^{\infty,\mathbb O_{\mathcal{NE}}^\delta,\mathrm{nv}}(\rho_{AB})\\	\leq&
 	\inf_{\sigma_{AB}\in Sep}	D_{\max}(\sigma_{AB}\|\rho_{AB})
 		+
 	\inf_{\sigma_{AB}\in Sep}	D(\rho_{AB}\|\sigma_{AB}),
 	\end{align*}
 	here $E_{d,\mathrm{err},p}^{\infty,\mathbb O_{\mathcal{NE}}^\delta,\mathrm{nv}}(\rho_{AB})$ denotes the error exponents of probabilistic entanglement distillation protocols for $\rho_{AB}$ under $\mathbb{O}_{\mathcal{NE}}^{\delta}$ with success probabilities $\liminf\limits_{n\rightarrow\infty}p_n>0.$
 \end{theorem}

For the Werner states in Example \ref{e1}, by Theorem \ref{thepb0}, when $p\in (0,\frac{1}{2})$,
\begin{align*}
	E_{d,\mathrm{err},p}^{\infty,\mathbb O_{\mathcal{NE}}^\delta,\mathrm{nv}}
(\rho_p)
\leq& D_{\max}(\rho_{\frac{1}{2}}||\rho)+D(\rho||\rho_{\frac{1}{2}})\\
=&
(1-p)\log\frac{1-p}{p}.
\end{align*}
On the other hand, the unrestricted postselected exponent is
\begin{align*}
	\widehat D_{\Omega,Sep}^{\mathrm{reg}}(\rho_p)
=
\log\frac{1-p}{p}.
\end{align*}
Since \(0<p<1/2\),
\begin{align*}
	(1-p)\log\frac{1-p}{p}
	<
	\log\frac{1-p}{p}.
\end{align*}

\emph{This shows that the nonvanishing-success requirement leads to a genuinely different asymptotic regime from unrestricted postselection.
}

Next we present the error exponents of the probabilistic entanglement distillation for a bipartite state under $\mathbb{O}_{\mathcal{DNE}}^{\delta}$ instruments.

\begin{theorem}\label{th2}
Assume $\rho_{AB}^{\otimes n}$ is a bipartite state on $\mathcal{H}_{AB}^{\otimes n}$. Let $\delta\in (0,1)$ and $\{m_n\ge 2\}_n$ be a sequence of natural numbers, then the error exponent of probabilistic entanglement distillation under $\mathbb{O}_{\mathcal{DNE}}^{\delta}$ satisfies
 \begin{align*}
	\hat{\beta}^{\mathbb{SEP}}_{\frac{\delta+2}{m_n+1},Sep}(\rho_{AB}^{\otimes n})- \log\frac{m_n+1}{m_n}\ge&	nE_{d,err,p}^{(m_n),\mathbb{O}_{\mathcal{DNE}}^{\delta}}(\rho_{AB}^{\otimes n})\\\ge& \hat{\beta}_{\frac{\delta+1}{m_n},Sep}^{\mathbb{SEP}}(\rho_{AB}^{\otimes n}).
\end{align*}
Besides, the zero-rate conditional error exponent of probabilistic entanglement distillation under $\mathbb{O}_{\mathcal{NE}}^{\delta}$ for $\rho_{AB}$ is 
\begin{align*}
	\liminf\limits_{n\rightarrow\infty} E^{\infty,\mathbb{O}_{\mathcal{DNE}}^{\delta}}_{d,err,p}(\rho_{AB})=&\hat{D}_{\Omega,Sep}^{reg,\mathbb{SEP}}(\rho_{AB}).
\end{align*}
\end{theorem}

The proof of Theorem \ref{th2} is placed in the Appendix \ref{aa0}.

\subsection{The resource theory of coherence}
Quantum coherence springs from the state superposition principle, which distinguishes quantum information from classical highly \cite{streltsov2017colloquium}. Coherence is closely related to quantum entanglement \cite{horodecki2009quantum}, quantum nonlocality \cite{brunner2014bell}, quantum imaginarity \cite{Hickey_2018,wu2021resource,wu2024resource,PhysRevA.111.L050401} and so on. Coherence also plays crucial roles in quantum algorithms \cite{rastegin2018role,ahnefeld2022coherence,li2023evolution,ye2026coherence,zhou2026cpl}, transport theory \cite{witt2013stationary,shi2024coherence}, and biology \cite{plenio2008dephasing,huelga2013vibrations}. 

Let $\mathcal{H}$ be a $d$-dimensional Hilbert space and $\mathcal{E}=\{\ket{i}\}$ be the set consisting of a prescribed orthonormal basis of $\mathcal{H},$ the set of incoherent states $I_{\mathcal{E}}$ is composed of all the states that are diagonal with respect to the basis $\{\ket{i}\}$. A typical maximally coherent state can be written as $\ket{\psi_d}=\frac{1}{\sqrt{d}}\sum_{i=0}^{d-1}\ket{i}.$ One of the most common used  coherence quantifiers is generalized robustness of coherence, $R_{G}(\cdot)$ \cite{napoli2016robustness,piani2016robustness}, the generalized robustness of coherence for $\rho$ is defined as
\begin{align*}
	R_G(\rho)=\min\{s|\frac{\rho+s\sigma}{1+s}\in I_{\mathcal{E}}, \sigma\in \boldsymbol{D}({\mathcal{H}})\}.
\end{align*}

 Assume $\{\mathcal{E}_i\}$ is a set of subchannels such that $\sum_i\mathcal{E}_i$ is a channel, then the $\delta$-approximately $\mathcal{NG}$ instruments is defined as
\begin{align*}
	\mathbb{O}_{\mathcal{NG}}^{\delta}=\{\mathcal{E}|R_G(\frac{\mathcal{E}_i(\sigma)}{\mathrm{tr}\mathcal{E}_i(\sigma)})\le \delta, \forall \mathcal{E}_i\in \mathcal{E},\sigma\in I_{\mathcal{E}}\}.
\end{align*}
where $R_G(\sigma)$ is the generalized robustness of coherence. 
When the dimension of target maximally coherent state is fixed, the probabilitic error exponent for $\rho$ in terms of $n$-blocklength scenario under $\mathbb{O}^{\delta}_{\mathcal{NG}}$-instruments $\{\mathcal{E}_i\}$ is defined as follows,
\begin{align*}
	E_{d,err,p}^{(m),\mathbb{O}^{\delta}_{\mathcal{NG}}}(\rho^{\otimes n})=&\sup-\frac{1}{n}\log\epsilon_n\\
	\textit{s. t.}\hspace{4mm}&	F(\frac{\mathcal{E}_i(\rho^{\otimes n})}{\mathrm{tr}(\mathcal{E}_i(\rho^{\otimes n}))},\Psi_m)\ge 1-\epsilon_n,\\
	&\exists \mathcal{E}_i\in\mathcal{E},\mathcal{E}\in \mathbb{O}^{\delta}_{\mathcal{NG}}.
\end{align*}
where $\Psi_m=\ket{\psi_m}\bra{\psi_m}$, and the supermum takes over all instruments $\mathcal{E}=\{\mathcal{E}_i\}\in \mathbb{O}_{\mathcal{NG}}^{\delta}.$ Moreover,
the zero-rate conditional error exponent of probabilistic coherent distillation under $\mathbb{O}_{\mathcal{NG}}^{\delta}$ for $\rho$ is 
\begin{align*}
	E^{\infty,\mathbb{O}_{\mathcal{NG}}^{\delta}}_{d,err,p}(\rho):=\liminf\limits_{n\rightarrow\infty} \sup E_{d,err,p}^{(m_n),\mathbb{O}^{\delta}_{\mathcal{NG}}}(\rho),
\end{align*}
where $\{m_n\ge2\}_n$ takes over all sequences with $\lim\limits_{n\rightarrow\infty}\frac{\log m_n}{n}=0.$

Next we will present an analytical formula for the probabilistic coherence distillation under $\mathbb{O}^{\delta}_{\mathcal{NG}}$ instruments. 
\begin{theorem}\label{th3}
	Assume $\rho^{\otimes n}$ is a state on $\mathcal{H}^{\otimes n}$. Let $m\ge2\in \mathbb{N}$, $\delta\in (0,1)$. Then the conditional error exponent of the coherence distillation under $\mathbb{O}^{\delta}_{\mathcal{NG}}$ is 
	\begin{align*}
		E_{d,err,p}^{(m),\mathbb{O}^{\delta}_{\mathcal{NG}}}(\rho^{\otimes n} )\le\frac{1}{n}\hat{\beta}_{\frac{\delta+1}{m},ic}(\rho^{\otimes n}).
	\end{align*}
	
Besides, the zero-rate conditional error exponent of probabilistic coherence distillation under $\mathbb{O}_{\mathcal{NG}}^{\delta}$ for $\rho$ is
	\begin{align*}
		E_{d,err,p}^{\infty,\mathbb{O}^{\delta}_{\mathcal{NG}}}(\rho)\le\hat{D}_{\Omega,ic}^{reg }(\rho).
	\end{align*}
\end{theorem}
The proof of Theorem \ref{th3} is placed in the Appendix \ref{apc}.

\subsection{The resource theory of magic}
  According to the Gottesman-Knill theorem \cite{Scott2004}, a quantum circuit comprised of only Clifford gates can be simulated efficiently on a classical computer. Hence, to reflect the advantages in comparison with classical computers, a quantum computer should be capable of non-Clifford operations or have nonstabilizer states available to it. Such states and operations can be seen as resources to perform universal quantum computation \cite{BravyiKitaev2005,Veitch_2014,HowardCampbell2017}. Similar with the resources of quantum entanglement and coherence, a key problem is to consider the transformation from the nonstabilizer states into high quality nonstabilizer states \cite{CampbellBrowne2010,CampbellAnwarBrowne2012,BravyiHaah2012,HaahEtAl2017,HaahHastings2018,WangWildeSu2020,Wills2025}. 

Here we consider the Hilbert state $\mathcal{H}$ with $\dim\mathcal{H}=3$. For the magic resource of the qutrit system, the Strange state is one of the class of states with the maixmal sum negativity of discrete Wigner function \cite{Veitch_2014}, which can be written as 
\begin{align*}
	\ket{S}=\frac{1}{\sqrt{2}}(\ket{1}-\ket{2}).
\end{align*}

Next we present conditional error exponents of the following probabilistic magic transformation from $\rho$ to $\mathbb{S}$ under $\mathbb{O}_{\mathcal{NG}}^{\delta}$ instruments.

\begin{theorem}\label{edm}
	Assume $\rho$ is a state on $\mathcal{H}$, let $\delta\in (0,1)$, then the conditional error exponent of the probabilistic magic distillation under $\mathbb{O}^{\delta}_{\mathcal{NG}}$ is 
\begin{align*}
	E_{d,err,p}^{\infty,\mathbb{O}^{\delta}_{\mathcal{NG}}}(\rho)=\hat{D}_{\Omega,STAB}^{reg }(\rho),
\end{align*}
	here $$\mathbb{O}^{\delta}_{\mathcal{NG}}=\{\mathcal{E}=\{\mathcal{E}_i\}|R_{\boldsymbol{G}}(\frac{\mathcal{E}_i(\sigma)}{\mathrm{ tr}\mathcal{E}_i(\sigma)})\le \delta,\forall\sigma\in STAB,\forall i\}.$$
\end{theorem} 

The proof of Theorem \ref{edm} is placed in the Appendix \ref{ptm}. As the magic resource satisfies the Brandao-Plenio axioms, then
based on Theorem \ref{edm}, Lemma \ref{mf6} and Theorem 14 in \cite{lami2026}, we have
\begin{corollary}
	Assume $\rho$  is a state, let $\delta\in (0,1)$, then
	\begin{align*}
		E_{d,err}^{\infty,\mathcal{NG}_\delta}(\rho)=D(STAB||\rho),
	\end{align*}
 where $$E_{d,err}^{\infty,\mathcal{NG}_\delta}(\rho)=\liminf\limits_{n\rightarrow\infty}\{\sup\limits_{\mathcal{E}\in \mathcal{NG}_{\delta}} -\frac{1}{n}\log\epsilon_n|	F(\mathcal{E}(\rho^{\otimes n}),\mathbb{S})\ge 1-\epsilon_n\},$$ and $	D(STAB||\rho)=\inf_{\sigma\in STAB} D(\sigma||\rho).$

\end{corollary}
\begin{example}\label{e3}
	Let
\begin{equation}
	\ket{\Psi}
	=
	\frac{1}{\sqrt{6}}
	\left(2\ket{0}-\ket{1}-\ket{2}\right)
\end{equation}
be the qutrit Norrell state, and define
\begin{equation}
	\rho_p
	=
	p\Psi+(1-p)\frac{\mathbb{I}}{3},
	\qquad
	\Psi=\ket{\Psi}\!\bra{\Psi},
	\qquad
	0\leq p\leq 1.
\end{equation}
For $2/5<p<1$, let $y_p>1$ be the unique solution of
\begin{equation}
	y_p^2(2y_p+1)
	=
	\frac{1+2p}{1-p}.
\end{equation}
Then 
\begin{align*}
	E_{d,err,p}^{\infty,\mathbb{O}_{\mathcal{NG}}^{\delta}}(\rho)=&D_{\Omega,STAB}(\chi_p)\\=&\begin{cases}
		0\hspace{18mm} p\in [0,\frac{2}{5}]\\
		\log\frac{1+2p}{3(1-p)}\hspace{5mm}p\in (\frac{2}{5},1)
	\end{cases}.
\end{align*}
\begin{align*}
		&D(\operatorname{STAB}_1\Vert\rho_p)\\	=&	\begin{cases}
			0,
			&
			0\leq p\leq \dfrac{2}{5},
			\\[3mm]
			\displaystyle
			q_\Psi\log\frac{q_\Psi}{a_p}
			+
			q_u\log\frac{q_u}{b_p}
			+
			q_v\log\frac{q_v}{b_p},
			&
			\dfrac{2}{5}<p<1.
		\end{cases}
\end{align*}
	where $q_\Psi=\frac{2y_p+1}{3y_p+2},q_u=\frac{y_p}{3y_p+2},q_v=\frac{1}{3y_p+2},a_p=\frac{1+2p}{3},$ $b_p=\frac{1-p}{3}.$
\end{example}

The proof of Example \ref{e3} is placed in the Appendix \ref{ptm}. In Fig. \ref{fig3}, we plot the asymptotic error exponent of magic distilltion under resource nongenerating operations. The red line is on the asymptotic error exponent under probabilistic scenario, while the blue line is on the success probability equaling 1. From Fig. \ref{fig3}, we have postselection can strictly improve the conditional error exponent of the (probabilistic) distillation over the deterministic counterpart.

\begin{figure}[h] 
	\centering
	\includegraphics[width=0.5\textwidth]{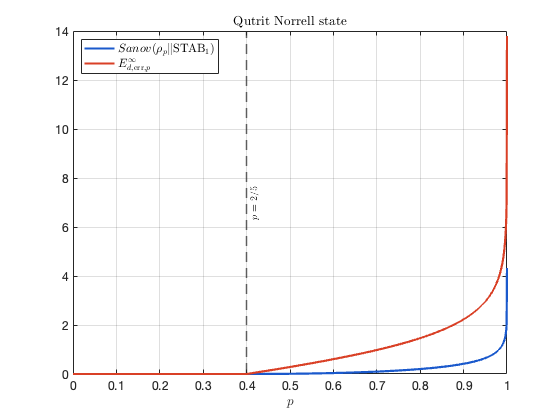} 
	\caption{The asymptotic error exponent of (probabilistic) magic distillation of $\rho_p$ under resource nongenerating instruments(operations). }
	\label{fig3}
\end{figure}
\section{Conclusion}\label{con}

Our results identify a common information-theoretic mechanism behind probabilistic quantum resource distillation. When a protocol is allowed to discard unfavorable outcomes and condition on success, its error exponent is no longer described by the same structure as in deterministic transformations. Instead, the normalization inherent in postselection leads naturally to postselected hypothesis testing and to scale-invariant, projective notions of distinguishability. This provides an explanation for why probabilistic and deterministic distillation can exhibit different error exponents.
We demonstrated this mechanism across three representative quantum resources. For entanglement and magic, the conditional error exponent at zero distillation rate admits an exact asymptotic characterization in terms of regularized projective distinguishability from the corresponding free-state sets. For coherence, the same framework yields general finite-blocklength and asymptotic bounds. The fact that the same postselected testing structure appears in these distinct resource theories suggests that it is not tied to the geometry of a particular free-state set, but rather reflects a general feature of conditioning in probabilistic quantum transformations.

The explicit examples further show that postselection can produce a strictly larger conditional error exponent than deterministic distillation. This separation should, however, be interpreted carefully. The improvement concerns the quality of the output conditioned on a successful event, rather than the overall performance of the protocol. Furthermore, we also presented an example showing that the conditional error exponents for probabilistic entanglement distillation under a strictly positive success probability constraint are smaller than those with no constraints on the success probability..

Several other questions remain open. It would be useful to determine when the regularized projective quantities arising here admit single-letter expressions, to understand whether analogous structures govern finite-rate resource distillation, and to extend the approach to other probabilistic quantum transformation tasks. More generally, the connection established here suggests that postselection should be viewed not simply as a way of improving selected outcomes, but as a change in the underlying distinguishability problem. In this sense, probabilistic quantum processing possesses a notion of reliability that is structurally different from that of deterministic protocols.

  \section{Acknowledgement}
 The author used generative AI, ChatGPT 5.6 Sol, to prove Lemma 22 and solve Examples 4, 9, and 23, as well as to prepare the manuscript. 
X. S. was supported by the National Natural Science Foundation of China (Grant No. 12301580).
\bibliographystyle{IEEEtran}
\bibliography{ref}

\setcounter{equation}{0}

\clearpage
\onecolumngrid
\setcounter{page}{1}
\setcounter{equation}{0}
\setcounter{figure}{0}
\renewcommand{\theequation}{S\arabic{equation}}
\renewcommand{\thefigure}{S\arabic{figure}}

\setcounter{secnumdepth}{2}
\setcounter{tocdepth}{2}
\appendix

\section{Appendix}\label{app}

\subsection{Quantum Relative Entropies}
Assume $\rho$ and $\sigma$ are two states, let $\alpha\in (1,\infty]$, then the $\alpha$-sandwiched Renyi divergence $\tilde{D}_{\alpha}(\rho,\sigma)$ for $\rho$ and $\sigma$ is defined as 
\begin{equation*}
	\tilde{D}_{\alpha}(\rho,\sigma)=
	\begin{cases}
		\frac{\alpha}{\alpha-1}\log||\sigma^{\frac{1-\alpha}{2\alpha}\rho\sigma^{\frac{1-\alpha}{2\alpha}}}||_{\alpha}\hspace{3mm} \textit{if $supp(\rho)\subseteq supp(\sigma)$,}\\
		+\infty \hspace{5mm} \textit{otherwise},
	\end{cases}
\end{equation*}
when $\alpha\rightarrow1$, $\tilde{D}_{\alpha}(\rho,\sigma)$ tends to the quantum relative entropy of $\rho$ and $\sigma$, ${D}(\rho||\sigma)=\mathrm{tr}[\rho(\log\rho-\log\sigma)].$ 

Next we define the quantum maxi-relative entropy for two states $\rho$ and $\sigma$ with $supp(\rho)\subseteq supp(\sigma)$,
\begin{align}
	D_{max}(\rho,\sigma)=\log\inf& \hspace{2mm}\lambda\label{dpmax}\\
	\textit{s. t.}\hspace{5mm}&\rho\le\lambda\sigma\nonumber\\
	&\lambda\in \mathbb{R}^{+},\nonumber
\end{align}
otherwise, $D_{max}(\rho,\sigma)$ tends to the infty. 
The dual program of $(\ref{dpmax})$ is 
\begin{align}
	D_{max}(\rho,\sigma)=\log\max&\hspace{3mm}\mathrm{ tr}\rho X\label{ddmax}\\
	\textit{s. t.}\hspace{5mm}&\mathrm{ tr}\sigma X\le 1,\nonumber\\
	&X\ge0.\nonumber
\end{align}

The Hilbert projective metric between two states $\rho$ and $\sigma$ is
\begin{align*}
	D_{\Omega}(\rho,\sigma)=D_{max}(\rho,\sigma)+D_{max}(\sigma,\rho).\\
	\Omega(\rho,\sigma)=2^{D_{\Omega}(\rho,\sigma)}.
\end{align*} 

After defining the Hilbert projective metric between two states, it is natural to define the divergence between the two states after measurements $M.$ Assume $\mathbb{M}$ is a class of measurements, $$\mathbb{M}=\{(M_i)|M_i\ge 0,\sum_iM_i=\mathbb{I}, M_i\in \mathcal{T}\},$$
here $\mathcal{T}$ is a convex set of nonnegative operations, the $\mathbb{M}$-Hilbert projective metric between $\rho$ and $\sigma$, $D_{\Omega,\mathbb{M}}(\rho,\sigma),$ is defined as
\begin{align*}
	D^{\mathbb{M}}_{\Omega}(\rho,\sigma)=\sup_{M\in\mathbb{M}}D_{\Omega}(\mathcal{M}(\rho),\mathcal{M}(\sigma)),
\end{align*}
where the supermum $M=\{M_i\}_i$ takes over all the measurements in $\mathbb{M}$, and $\mathcal{M}(\cdot)=\sum_i\mathrm{tr}(M_i\cdot)\ket{i}\bra{i}.$
The method to generalize the divergences based on measurements is highly helpful to the quantum information theory.

Next we present the following properties of $D_{\Omega}(\rho,\sigma)$.
\begin{lemma}\label{l1}\cite{regula2022tight}
	Assume $\rho$ and $\sigma$ are two states, then
	\begin{itemize}
		\item[(1.)] $D_{\Omega}(\rho,\sigma)\ge 0$, and the quality happens if and only if $\rho=\sigma$.
		\item[(2.)] $D_{\Omega}(\rho,\sigma)=D_{\Omega}(\sigma,\rho)$.
		\item[(3.)] For arbitrary positive numbers $\lambda$ and $\varphi$, then $D_{\Omega}(\rho,\sigma)=D_{\Omega}(\lambda\rho,\varphi\sigma).$
		\item[(4.)] The quantity $D_{\Omega}(\cdot,\cdot)$ satisifies the data-processing property under the positive map, that is, for each positive linear map $\mathcal{E}$, 
		\begin{align*}
			D_{\Omega}(\mathcal{E}(\rho),\mathcal{E}(\sigma))\le D_{\Omega}(\rho,\sigma).
		\end{align*}
		\item[(5.)] $D_{\Omega}(\rho,\sigma)$ can be computed under the semidefinite programming method, 
		\begin{align}
			D_{\Omega}(\rho,\sigma)=&\log\sup \mathrm{ tr}A\rho\label{f1}\\
			\textit{s. t.} \hspace{4mm}& \mathrm{ tr}B\rho=1,\nonumber\\
			&\mathrm{ tr}(B-A)\sigma\ge0,\nonumber\\
			&A,B\ge 0\nonumber
		\end{align}
		\item[(6.)] Assume $\rho$ and $\sigma$ are two states, $D_{\Omega}(\rho^{\otimes n},\sigma^{\otimes n})=nD_{\Omega}(\rho,\sigma)$.
		\item[(7.)] Assume $\mathbb{M}$ is a class of measurements, $$\mathbb{M}=\{(M_i)|M_i\ge 0,\sum_iM_i=\mathbb{I}, M_i\in \mathcal{T}\},$$here $\mathcal{T}$ is a convex set of nonnegative operations,  then
		\begin{align}
			D^{\mathbb{M}}_{\Omega}(\rho,\sigma)=&\log\sup \mathrm{ tr}A\rho\label{f1}\\
			\textit{s. t.} \hspace{4mm}& \mathrm{ tr}B\rho=1,\nonumber\\
			&\mathrm{ tr}(B-A)\sigma\ge0,\nonumber\\
			&A,B\in cone(\mathcal{T}).\nonumber
		\end{align}
	\end{itemize}
\end{lemma}

\subsection{Quantum Postselected Hypothesis Testing}
Quantum state discrimination is a fundamental quantum information task. Recently, the authors in \cite{regula2023postselected} addressed the following problem. Assume Alice receives a state, and she knows that the state is $\rho$ or $\sigma$, her aim is to determine which state she obtained. In the scenario, she can perform a three-outcome positive operator-valued measure (POVM), $M=\{M_1,M_2,M_0\}.$ The outcome 1 and 2 correspond to the state $\rho$ and $\sigma$, respectively, while the outcome 0 corresponds to the inconclusive. Let
\begin{align*}
	\textit{conditional type I error:}\hspace{5mm}\overline{\alpha}(M)=\frac{\mathrm{ tr}M_2\rho}{\mathrm{ tr}(M_1+M_2)\rho},\\
	\textit{conditional type II error:}\hspace{5mm}\overline{\beta}(M)=\frac{\mathrm{ tr}M_1\sigma}{\mathrm{ tr}(M_1+M_2)\sigma},
\end{align*}
Assume $\mathcal{F}$ is a convex and closed set of quantum states, and the postselected hypothesis testing between a state $\rho$ and the set $\mathcal{F}$ is 

\begin{align}
	\overline{\beta}_{\epsilon,\mathcal{F}}(\rho)=-\log\inf_{M\in \mathcal{M}_3}\{\sup_{\sigma\in\mathcal{F}}\frac{\mathrm{ tr}M_1\sigma}{\mathrm{ tr}(M_1+M_2)\sigma}|\frac{\mathrm{ tr}M_2\rho}{\mathrm{ tr}(M_1+M_2)\rho}\le\epsilon\}\label{pht}
\end{align}
where $M$ takes over all the elements in $\mathcal{M}_3$, and $\mathrm{ tr}(M_1+M_2)\sigma,\mathrm{ tr}(M_1+M_2)\rho> 0.$

\begin{lemma}\cite{regula2023postselected}\label{l2}
	Assume $\mathcal{F}$ is a convex and closed set of quantum states, then
	\begin{align*}
		\overline{\beta}_{\epsilon,\mathcal{F}}(\rho)=\frac{\epsilon}{1-\epsilon}\min_{\sigma\in\mathcal{F}}\Omega(\rho,\sigma)+1.
	\end{align*}
	Here $\Omega(\rho,\sigma)=2^{D_{\Omega}(\rho,\sigma)}.$ 
	
	When $\mathcal{F}$ is closed under the tensor operations,
	\begin{align*}
		\lim\limits_{n\rightarrow\infty}\frac{1}{n}\overline{\beta}_{\epsilon,\mathcal{F}}(\rho^{\otimes n})=\lim\limits_{n\rightarrow\infty}\frac{1}{n}\min_{\sigma_n\in\mathcal{F}}D_{\Omega}(\rho^{\otimes n},\sigma_n).
	\end{align*}
	
	Furthermore, the Hilbert projective metric satisfies the asymptotic equipartition property,
	\begin{align*}
		\lim\limits_{\epsilon\rightarrow 0}\lim\limits_{n\rightarrow\infty}\frac{1}{n}D^{\epsilon}_{\Omega}(\rho^{\otimes n},\mathcal{F}):=&\lim\limits_{\epsilon\rightarrow 0}\lim\limits_{n\rightarrow\infty}\min_{\rho^{'}\in B_{\epsilon}(\rho^{\otimes n})}\frac{1}{n}D_{\Omega}(\rho^{`},\mathcal{F})\\
		=&D_{\mathcal{F}}^{\infty}(\rho).
	\end{align*}
	where the minimum in the first equality takes over all the states in  $B_{\epsilon}(\rho^{\otimes n})=\{\rho^{'}|\frac{1}{2}||\rho^{\otimes n}-\rho^{'}||_1\le \epsilon\},$ and $D_{\mathcal{F}}^{\infty}(\rho)$ in the second equality is defined as $D_{\mathcal{F}}^{\infty}(\rho)=\lim\limits_{n\rightarrow\infty}\frac{1}{n}\min_{\sigma_n\in\mathcal{F}_n}D(\rho^{\otimes n}||\sigma_n).$
\end{lemma}

Here we address a reversed problem of the composite postselected hypothesis testing. Assume $\mathcal{F}$ is a convex and compact set, $M\in\mathcal{M}_3$ is a feasible POVM, the conditional type $II$ error is defined as
\begin{align*}
	\overline{\beta}(M)=\frac{\mathrm{tr}M_1\rho}{\mathrm{ tr}(M_1+M_2)\rho},
\end{align*}
while the conditional type $I$ error $$\overline{\alpha}_{\mathcal{F}}(M)=\sup_{\sigma\in\mathcal{F}}\frac{\mathrm{tr}M_2\sigma}{\mathrm{tr}(M_1+M_2)\sigma}.$$
The reversed composite postselected hypothesis testing, $\hat{\beta}_{\epsilon,\mathcal{F}}(\rho)$, is defined as follows,
\begin{align}
	\hat{\beta}_{\epsilon,\mathcal{F}}(\rho)=&-\log\inf_{M\in\mathcal{M}_3}\frac{\mathrm{tr}M_1\rho}{\mathrm{ tr}(M_1+M_2)\rho}\label{lf0}\\
	\textit{s. t.}&\hspace{4mm} \frac{\mathrm{tr}M_2\sigma}{\mathrm{tr}(M_1+M_2)\sigma}\le\epsilon, \forall\sigma\in\mathcal{F},\nonumber\\
	&\hspace{4mm}0\le M_1+M_2\le \mathbb{I}.\nonumber
\end{align}
When $M_1+M_2=\mathbb{I}$, the postselected hypothesis testing turns into the quantum hypothesis testing between $\mathcal{F}$ and $\rho$, $D_{H}^{\epsilon}(\mathcal{F}||\rho)$.
Furthermore, when $\mathcal{M}_3$ in (\ref{lf0}) is in a class of $\mathbb{M}$, then we define $\hat{\beta}_{\epsilon,\mathcal{F}}^{\mathbb{M}}(\rho)$ as follows
\begin{align}
\hat{\beta}^{\mathbb{M}}_{\epsilon,\mathcal{F}}(\rho)=&-\log\inf_{M\in\mathcal{M}_3}\frac{\mathrm{tr}M_1\rho}{\mathrm{ tr}(M_1+M_2)\rho}\label{lfm}\\
\textit{s. t.}&\hspace{4mm} \frac{\mathrm{tr}M_2\sigma}{\mathrm{tr}(M_1+M_2)\sigma}\le\epsilon, \forall\sigma\in\mathcal{F},\nonumber\\
&\hspace{4mm}0\le M_1+M_2\le \mathbb{I}, \mathcal{M}_3\in \mathbb{M}.\nonumber
\end{align}

The analytical formula of $\hat{\beta}_{\epsilon,\mathcal{F}}(\rho)$ is presented in the following corollary.

\begin{corollary}\label{c1}
	Assume $\mathcal{F}$ is a convex and compact set of quantum states, then 
	\begin{align}
		{\hat{\beta}_{\epsilon,\mathcal{F}}(\rho)}
		=\log[1+\frac{\epsilon}{1-\epsilon}{\hat{\Omega}}_{\mathcal{F}}(\rho)]
	\end{align}
	
	When each family set $(\mathcal{F}_n)_n$ are convex and compact, and $(\mathcal{F}_n)_n$ is closed under tensor product, we have
	\begin{align}
		\lim\limits_{n\rightarrow \infty}\frac{1}{n}{\hat{\beta}_{\epsilon,\mathcal{F}}(\rho^{\otimes n})}=\hat{D}_{\Omega,\mathcal{F}}^{reg}(\rho):=
		\lim\limits_{n\rightarrow\infty}\frac{1}{n}\log\hat{\Omega}_{\mathcal{F}}(\rho^{\otimes n}).
	\end{align}
\end{corollary}
\begin{proof}
	Here we take a similar method in \cite{regula2023postselected} to show the theorem. Based on the definition of $\hat{\beta}_{\epsilon,\mathcal{F}}(\rho),$ we have
	
	\begin{align*}
		&\hat{\beta}_{\epsilon,\mathcal{F}}(\rho)\\
		=&-\log\inf_{M\in \mathcal{M}_3}\{\frac{\mathrm{tr}M_1\rho}{\mathrm{tr}(M_1+M_2)\rho}|\frac{\mathrm{tr}M_2\sigma}{\mathrm{tr}(M_1+M_2)\sigma}\le\epsilon, \forall\sigma\in\mathcal{F}, 0\le M_1+M_2\le \mathbb{I}\}\\
		=&-\log\inf_{t,M\in \mathcal{M}_3}\{t|\frac{\mathrm{tr}M_1\rho}{\mathrm{tr}(M_1+M_2)\rho}\le t,\frac{\mathrm{tr}M_2\sigma}{\mathrm{tr}(M_1+M_2)\sigma}\le\epsilon, \forall\sigma\in\mathcal{F}, 0\le M_1+M_2\le \mathbb{I}\}\\
		=&-\log\inf_{\tilde{t},M_1,M_2^{'}\ge 0}\{\frac{1}{\tilde{t}}|\frac{\mathrm{tr}M_2^{'}\rho}{\mathrm{tr}M_1\rho}\ge 1,\frac{\mathrm{ tr}M_1\sigma}{\mathrm{ tr}M_2^{'}\sigma}\ge \frac{1-\epsilon}{\epsilon}(t^{'}-1),\forall\sigma\in \mathcal{F}\},
	\end{align*}
	
	In the third equality, we denote $M_2^{'}=\frac{t }{1-t}M_2$, in the last equality, $\tilde{t}=\frac{1}{t}$. Then we have
	\begin{align*}
		&2^{\hat{\beta}_{\epsilon,\mathcal{F}}(\rho)}\\
		=&\inf_{\sigma\in\mathcal{F}}\sup_{\tilde{t}\ge 0,M_1,M_2^{'}\ge 0}\{\tilde{t}|\frac{\mathrm{ tr}M_1\rho}{\mathrm{ tr}M_2^{'}\rho}\le 1,\frac{\mathrm{ tr}M_1\sigma}{\mathrm{ tr}M_2^{'}\sigma}\ge \frac{1-\epsilon}{\epsilon}(t^{'}-1)\}\\
		=&\inf_{\sigma\in \mathcal{F}}\sup_{\tilde{t}\ge 0,M_1,M_2^{'}\ge 0}\{   1+\frac{\epsilon}{1-\epsilon}\frac{\mathrm{ tr}M_1\sigma}{\mathrm{ tr}M_2^{'}\sigma}|\frac{\mathrm{ tr}M_1\rho}{\mathrm{ tr}M_2^{'}\rho}\le 1\}.
	\end{align*}
	As 
	\begin{align*}
		\min_{\rho\in \mathcal{F}}\Omega(\rho,\sigma)=&\sup_{A,B}\{\frac{\mathrm{tr}A\rho}{\mathrm{tr}B\rho}|\frac{\mathrm{tr}A\sigma}{\mathrm{ tr}B\sigma}\le 1,\forall\rho\in \mathcal{F}\}\\
		=&\inf_{\rho\in\mathcal{F}}\sup_{A,B}\{\frac{\mathrm{tr}A\rho}{\mathrm{tr}B\rho}|\frac{\mathrm{tr}A\sigma}{\mathrm{ tr}B\sigma}\le 1\},
	\end{align*}
	then
	\begin{align*}
		2^{\hat{\beta}_{\epsilon,\mathcal{F}}(\rho)}
		=&1+\frac{\epsilon}{1-\epsilon}\min_{\sigma\in\mathcal{F}}\Omega(\sigma,\rho)
	\end{align*}
	
For a generic $n$, let $\sigma_n\in\mathcal{F}_n$ be the optimal for $\rho^{\otimes n}$ in terms of $\Omega(\cdot,\rho^{\otimes n})$, 
	\begin{align}
		&	\frac{\epsilon}{1-\epsilon}\Omega(\sigma_n,\rho^{\otimes n})\le 2^{\hat{\beta}_{\epsilon,\mathcal{F}}(\rho^{\otimes n})}\le \frac{1}{1-\epsilon}\Omega(\sigma_n,\rho^{\otimes n})\nonumber\\
		\Longrightarrow&
		\log\frac{\epsilon}{1-\epsilon}+\log\Omega(\sigma_n,\rho^{\otimes n})\le {\hat{\beta}_{\epsilon,\mathcal{F}}(\rho^{\otimes n})}\nonumber\\\le& \log\frac{1}{1-\epsilon}+\log\Omega(\sigma_n,\rho^{\otimes n}),\label{lf00}
	\end{align}
	next when $\mathcal{F}_n$ is closed under tensor product, $\log\hat{\Omega}_{\mathcal{F}_{m+n}}(\rho^{\otimes m+n})\le \log\hat{\Omega}_{\mathcal{F}_{m}}(\rho^{\otimes m})+\log\hat{\Omega}_{\mathcal{F}_{n}}(\rho^{\otimes n})$, due to Fekete's Lemma, $\lim\limits_{n\rightarrow\infty}\frac{1}{n}\log\hat{\Omega}_{\mathcal{F}_{n}}(\rho^{\otimes n})$ exists. Then dividing n to both sides of (\ref{lf00}) and taking the limit, we have
	\begin{align*}
		\lim\limits_{n\rightarrow\infty}\frac{1}{n}\hat{\beta}_{\epsilon,\mathcal{F}}(\rho^{\otimes n})=\lim\limits_{n\rightarrow\infty}\frac{1}{n}\log\hat{\Omega}_{\mathcal{F}}(\rho^{\otimes n}).
	\end{align*}
\end{proof}

	\begin{corollary}\label{c2}
	Assume $\mathcal{F}$ is a convex and compact set of quantum states on $\mathcal{H}$ with $\frac{I}{d}\in \mathcal{F}$, here $d$ is the dimension of $\mathcal{H},$ $\mathbb{M}$ is a class of measurements on $\mathcal{H}$ with $(\mathbb{M}^{*})^{*}=cone(\mathbb{M})$, then 
	\begin{align}
		{\hat{\beta}^{\mathbb{M}}_{\epsilon,\mathcal{F}}(\rho)}
		=\log[1+\frac{\epsilon}{1-\epsilon}{\hat{\Omega}}^{\mathbb{M}}_{\mathcal{F}}(\rho)],
	\end{align}
	here 
	\begin{align*}
		\hat{\Omega}^{\mathbb{M}}_{\mathcal{F}}(\rho)=&\inf   \gamma\\
		\textit{s. t.}\hspace{3mm}&\rho\preceq_{\mathbb{M}^{*}}\tilde{\sigma}\preceq_{\mathbb{M}^{*}}\gamma\rho,\\
		&	\tilde{\sigma}\in cone(\mathcal{F}),
	\end{align*}
	here $X\preceq_{\mathbb{M}^{*}}Y$ means that $\mathrm{ tr}(Y-X)H\ge 0,$ $\forall H\in M$ and $M\in \mathbb{M}$.

	When the family set $(\mathbb{M}_n)_n$ and $(\mathcal{F}_n)_n$ are convex and compact, $(\mathbb{M}_n)_n$ and $(\mathcal{F}_n)_n$ are closed under tensor product, we have
	\begin{align}
		\liminf\limits_{n\rightarrow \infty}\frac{1}{n}{\hat{\beta}^{\mathbb{M}}_{\epsilon,\mathcal{F}}(\rho^{\otimes n})}=\hat{D}_{\Omega,\mathcal{F}}^{reg,\mathbb{M}}(\rho):=
		\liminf\limits_{n\rightarrow\infty}\frac{1}{n}\log\hat{\Omega}^{\mathbb{M}}_{\mathcal{F}}(\rho^{\otimes n}).
	\end{align}
	
\end{corollary}

First we present the dual form of $\Omega_{\mathcal{F}}^{\mathbb{M}}(\rho)$, which can be obtained through standard Lagrange duality arguments \cite{ponstein2004approaches},

\begin{align}
	\hat{\Omega}^{\mathbb{M}}_{\mathcal{F}}(\rho)=&\sup \mathrm{ tr}A\rho \label{dmhp}\\
	\textit{s. t.} \hspace{4mm}& \mathrm{ tr}B\rho=1,\nonumber\\
	&\mathrm{ tr}(B-A)\sigma\ge0,\nonumber\\
	&A,B\in\{\omega|\omega=\sum_i\mu_i M_i, \mu_i \ge 0\},\nonumber\\
	&M=\{M_i\}_i\in  \mathbb{M}.\nonumber
\end{align}

The strong duality can be proved by Slater's theorem \cite{rockafellar1997convex}, when taking $B=\mathbb{I}$ and $A=\epsilon \mathbb{I}$ for $\epsilon\in (0,1)$, it is feasible for the dual.

Based on a similar method of the proof of Corollary \ref{c1}, we have
	\begin{align*}
	2^{\hat{\beta}^{\mathbb{M}}_{\epsilon,\mathcal{F}}(\rho)}
	=&1+\frac{\epsilon}{1-\epsilon}\hat{\Omega}_{\mathcal{F}}^{\mathbb{M}}(\rho)
\end{align*}

For a generic $n$, let $\sigma_n\in\mathcal{F}_n$ be the optimal for $\rho^{\otimes n}$ in terms of $\Omega(\cdot,\rho^{\otimes n})$, 
\begin{align}
	&	\frac{\epsilon}{1-\epsilon}\Omega^{\mathbb{M}}(\sigma_n,\rho^{\otimes n})\le 2^{\hat{\beta}^{\mathbb{M}}_{\epsilon,\mathcal{F}}(\rho^{\otimes n})}\le \frac{1}{1-\epsilon}\Omega^{\mathbb{M}}(\sigma_n,\rho^{\otimes n})\nonumber\\
	\Longrightarrow&
	\log\frac{\epsilon}{1-\epsilon}+\log\Omega^{\mathbb{M}}(\sigma_n,\rho^{\otimes n})\le {\hat{\beta}^{\mathbb{M}}_{\epsilon,\mathcal{F}}(\rho^{\otimes n})}\le \log\frac{1}{1-\epsilon}+\log\Omega^{\mathbb{M}}(\sigma_n,\rho^{\otimes n}). \label{lf01}
\end{align}
Then dividing $n$ to both sides of (\ref{lf01}) and taking the limit, we have
\begin{align*}
	\liminf\limits_{n\rightarrow\infty}\frac{1}{n}\hat{\beta}^{\mathbb{M}}_{\epsilon,\mathcal{F}}(\rho^{\otimes n})=\liminf\limits_{n\rightarrow\infty}\frac{1}{n}\log\hat{\Omega}^{\mathbb{M}}_{\mathcal{F}}(\rho^{\otimes n}).
\end{align*}

\subsection{Proof of Theorem \ref{th1}}\label{ath1}

\begin{lemma}\label{gmm}
	Assume $\mathcal{H}$ is a Hilbert space with $dim(\mathcal{H})=d\ge2,$ both $M$ and $N_i$ $(i=1,2,\cdots,m)$ are semidefinite positive operators acting on $\mathcal{H}_{AB}$ with $M+\sum_i N_i\le \mathbb{I}_{AB}$, let $$\Lambda(X)=\mathrm{ tr}MX\cdot\mathbb{I}_0+\sum_{i=1}^m\mathrm{ tr}N_iX\cdot \frac{\mathbb{I}_i}{d_{i}},$$ here $\mathbb{I}_0=\Psi,$ $\mathbb{I}_i$ is the projector onto some subspace of $\mathcal{H}$,  $\mathbb{I}_i\perp\mathbb{I}_j$, $\forall i\ne j,$ and $d_i=dim(\mathbb{I}_i).$ Then for any $\epsilon>0$,   $$\sup\limits_{X\in \mathbf{F}} \frac{\max(\mathrm{tr}MX,\frac{\mathrm{tr}N_iX}{d_i})}{\mathrm{tr}MX+\sum_i \mathrm{tr}N_iX}\le \frac{\epsilon+1}{d}\Longrightarrow\Lambda(\cdot)\in \mathbb{O}_{\mathcal{NG}}^{\epsilon}.$$
\end{lemma}
\begin{proof}
	For any $X\in \mathbf{F},$
	\begin{align*}
		&\frac{\mathrm{ tr}\Lambda(X)}{d\max(\mathrm{tr}MX,\frac{\mathrm{tr}N_iX}{d_i})}[\frac{\Lambda(X)}{\mathrm{tr}\Lambda(X)}+\frac{(\max(\mathrm{ tr}MX,\frac{\mathrm{tr}N_iX}{d_i})-\mathrm{ tr}MX)\mathbb{I}_0+\sum_i[\max(\mathrm{ tr}MX,\frac{\mathrm{tr}N_iX}{d_i})-\frac{\mathrm{tr}N_iX}{{d_i}}] {\mathbb{I}_i}}{\mathrm{ tr}\Lambda(X)}]\\
		=& \frac{\mathbb{I}}{d}\in \mathbf{F},
	\end{align*}
	then we have $$R_{s,\mathbf{F}}(\Lambda(X))\le\frac{d\max(\mathrm{tr}MX,\frac{\mathrm{tr}N_iX}{d_i})}{\mathrm{ tr}\Lambda(X)}-1.$$ Hence, $\sup\limits_{X\in \mathbf{F}}R_{s,\mathbf{F}}(\mathrm{ tr}\Lambda(X))\le \epsilon$, $\Lambda(\cdot)$ is $\mathcal{NG}$.
\end{proof}

First we define the following twirling operation
\begin{align}
	\Lambda(\cdot)=\int_Gdg U_g(\cdot) U_g^{\dagger},\label{ff10}
\end{align}
where the integral is taken over the Haar measure of the group $G.$ Any unitary representation $\{U_g\}_{g\in G}$ of a finite or a compact Lie group $G$ can be decomposed into a direct sum of irreducible representation \cite{Fulton2013}. Accordingly, the decomposition of the representation $\{U_g\}$ is associated with the decomposition of the Hilbert space 
\begin{align*}
	\mathcal{H}=\bigoplus_{\mu} \mathcal{H}_{\mu}^{(1)}\otimes\mathcal{H}_{\mu}^{(2)},
\end{align*}
where $\mu$ is the irreducible representation, $\mathcal{H}_{\mu}^{(1)}$ is the subspace on which each irreducible represenation acts nontrivially, and $\mathcal{H}_{\mu}^{(2)}$ denotes the multiplicity subspace. Furthermore, as $dim \mathcal{H}_{\mu}^{(2)}=1,$ $\forall\mu$, each $U_g$ can be written as
\begin{align*}
	U_g=\bigoplus_{\mu} U_{\mu}^{(1)}(g).
\end{align*}
As for any $g\in G$,
\begin{align*}
	U_g\Lambda(\rho)U_g^{\dagger}=&\int_G dh U_gU_h\rho U_h^{\dagger}U_g^{\dagger}\\
	=&\int_G dhU_{gh}\rho U_{gh}^{\dagger}\\
	=&\int_G dg U_g\rho U_g^{\dagger}=\Lambda(\rho),
\end{align*}
Based on Schur's lemma \cite{Fulton2013}, each symmetric state $\sigma$ with $U_g\sigma U_g^{\dagger}=\sigma$ $\forall g,$ can be written as 
\begin{align*}
	\sigma=\bigoplus_{\mu}q_{\mu}\frac{\mathbb{I}_{\mu}^{(1)}}{d_{\mu}^{(1)}},
\end{align*}
here $\{q_{\mu}\}_{\mu}$ is a probability distribution, $\mathbb{I}_{\mu}^{(1)}$ is the projector onto $\mathcal{H}_{\mu}^{(1)},$ and $dim\mathcal{H}_{\mu}^{(1)}=d_{\mu}^{(1)}.$

Next define the projection onto invariant subspaces as 
\begin{align*}
	\mathcal{P}(\cdot)=\sum_{\mu}\mathbb{I}_{\mu}^{(1)}(\cdot)\mathbb{I}_{\mu}^{(1)}.
\end{align*}
Then for any state $\rho$, $\mathcal{P}(\rho)=\bigoplus_{\mu}p_{\mu}\sigma_{\mu}$, here $\sigma_{\mu}$ is a quantum state on $\mathcal{H}_{\mu}^{(1)}$. And
\begin{align*}
	\Lambda\circ\mathcal{P}(\rho)=\int_gdg U_g\mathcal{P}(\rho) U_g^{\dagger}=\int_gdg \mathcal{P}(U_g\rho U_g^{\dagger})=\Lambda(\rho), \hspace{2mm}\forall g\in G.
\end{align*}
Then 
\begin{align*}
	\Lambda(\rho)=\Lambda(\mathcal{P}(\rho))=&\bigoplus_{\mu}p_{\mu}\int_Gdg U_{\mu}^{(1)}(g)\sigma_{\mu} U_{\mu}^{(1)}(g)^{\dagger}\\
	=&\bigoplus_{\mu}p_{\mu}\frac{\mathbb{I}_{\mu}^{(1)}}{d_{\mu}^{(1)}}\hspace{3mm} ,
\end{align*}
Hence,
\begin{align}
	\Lambda(\rho)=&\sum_{\mu}p_{\mu}\frac{\mathbb{I}_{\mu}^{(1)}}{d_{\mu}^{(1)}}\nonumber\\
	=&\sum_{\mu}\frac{\mathbb{I}_{\mu}^{(1)}}{d_{\mu}^{(1)}}\mathrm{tr}[\rho \mathbb{I}_{\mu}^{(1)}].\label{ff11}
\end{align}

Based on the assumption, $\ket{\psi}$ is the eigenvector of any $U_g$. That is, $\ket{\psi}$ is in the invariant subspace of $U_g,$ $i.$$e.$ $U_g\ket{\psi}=e^{i\phi_g}\ket{\psi}$, and $\mathbb{I}_{\mu_1}=\Psi=\ket{\psi}\bra{\psi}.$ 
Combing with (\ref{ff11}), we have
\begin{align*}
	\Lambda(\rho)=\mathrm{tr}\rho\Psi\cdot\Psi+\sum_{\mu\ne \mu_1}\frac{\mathbb{I}_{\mu}^{(1)}}{d_{\mu}^{(1)}}\mathrm{tr}[\rho \mathbb{I}_{\mu}^{(1)}].
\end{align*}

\begin{lemma}\label{lf}
	Assume $\rho$ is a state, then its probabilistic distillation exponent for the reference state $\ket{\psi}$ under $\mathbb{O}_{\mathcal{NG}}^{\delta}$ instrument, $E_{d,err,p}(\rho)$, can be rewritten as
	\begin{align*}
		E_{d,err,p}(\rho)=&\sup-\log\epsilon\\
		\textit{s. t.}\hspace{4mm}&  \frac{\mathrm{tr}M\mathcal{E}_i(\rho)}{\mathrm{ tr}\mathcal{E}_i(\rho)}\ge 1-\epsilon,
		\mathcal{E}_i(X)=\mathrm{ tr}MX\cdot\Psi+\sum_{\mu\ne \mu_1}\mathrm{ tr}N_{\mu}X\cdot\frac{\mathbb{I}_{\mu}^{(1)}}{d_{\mu}^{(1)}},\\
		&M+\sum_{\mu\ne \mu_1}N_{\mu}\le \mathbb{I},M,N_{\mu}\ge 0,\forall\mu\ne \mu_1,\\
		&\mathcal{E}_i\in \mathcal{E},\mathcal{E}\in \mathbb{O}_{\mathcal{NG}}^{{\delta}}.
	\end{align*}
	
\end{lemma}
\begin{proof}
	When $\mathcal{E}_i\in \mathcal{E}$ and $\mathcal{E}\in \mathbb{O}_{\mathcal{NG}}^{\delta}$, $\Lambda\circ\mathcal{E}_i\in\Lambda\circ\mathcal{E}$, $\Lambda\circ\mathcal{E}\in \mathbb{O}_{\mathcal{NG}}^{\delta},$ then 
	\begin{align*}
		F(\frac{\mathcal{E}_i(\rho)}{\mathrm{tr}(\mathcal{E}_i(\rho))},\Psi)=&\bra{\psi}\frac{\mathcal{E}_i(\rho)}{\mathrm{tr}(\mathcal{E}_i(\rho))}\ket{\psi}\\
		=&\frac{\int dg\bra{\psi}U_g^{\dagger}\mathcal{E}_i(\rho)U\ket{\psi}dU}{\mathrm{ tr}\mathcal{E}_i(\rho)}\\
		=&\frac{\mathrm{ tr}\Psi(\mathrm{ tr}\mathcal{E}_i^{\dagger}(\Psi)\rho\cdot\Psi+\sum_{\mu\ne \mu_1}\frac{\mathbb{I}_{\mu}^{(1)}}{d_{\mu}^{(1)}}\mathrm{tr}[\rho \mathcal{E}_i^{\dagger}\mathbb{I}_{\mu}^{(1)}])}{\mathrm{ tr}\mathcal{E}_i(\rho )}\\
		=&\frac{\mathrm{ tr}\Psi\mathcal{E}_i(\rho)}{\mathrm{ tr}\mathcal{E}_i(\rho)}\ge 1-\epsilon_.
	\end{align*}
		Hence, we finish the proof.
\end{proof}

\emph{Theorem \ref{th1}} Assume $\rho$ is a state on $\mathcal{H}$. Let $\ket{\psi}$ be the reference state and $\delta> 0$, if there exists a finite or a compact Lie group $G$ with a unitary repreesntation $\{U_g\}_{g\in G}$ such that $U_g(\cdot)U_g^{\dagger}\in \mathbf{O}_{max}$, $\ket{\psi}$ is an eigenvector of any $U_g$, and the associated decomposition of $\mathcal{H}$ satisifies $(\ref{hs})$, then a lower bound of the conditional error exponent for the probabilistic transformation from $\rho$ to $\psi$ under $\mathbb{O}_{\mathcal{NG}}^{\delta}$ in terms of one-shot scenario is
\begin{align*}
	E_{d,err,p}(\rho)\ge \hat{\beta}_{\frac{\delta+1}{d},\mathbf{F}}(\rho).
\end{align*}

\begin{proof}
	Assume $\{(M,N_i,L)|M+\sum_i N_i+L=\mathbb{I},M,N_i\ge0\}$ is a POVM with $ \frac{\max(\frac{\mathrm{ tr}N_i\sigma}{d_i},\mathrm{tr}M\sigma)}{\mathrm{ tr}(M+\sum_i N_i)\sigma}\le \frac{{\delta}+1}{d}$ for any $\sigma\in \mathbf{F}.$ Next we construct the following subchannel $\mathcal{E}_1(\cdot)$, $\mathcal{E}_2(\cdot)$,
	\begin{align*}
		\mathcal{E}_1(\cdot)=&\mathrm{tr}M(\cdot)\Psi+\sum_i\mathrm{tr}N_i(\cdot)\frac{\mathbb{I}_i}{d_i},\\
		\mathcal{E}_2(\cdot)=&\mathrm{tr}L(\cdot)\frac{\mathbb{I}}{d}.
	\end{align*}
	For any $\sigma\in \mathbf{F}$, as $\frac{\mathbb{I}}{d}\in \mathbf{F}$, $\frac{\mathcal{E}_2}{\mathrm{tr}\mathcal{E}_2(\cdot)}\in\mathbf{O}_{max}$. Based on Lemma \ref{gmm}, $\{\mathcal{E}_1,\mathcal{E}_2\}\in \mathbb{O}_{\mathcal{NG}}^{\delta}$. At last, by Lemma \ref{lf},
	\begin{align*}
		E_{d,err,p}(\rho)\ge& -\log\frac{\sum_i\mathrm{ tr}N_i\rho}{\mathrm{tr}(M+\sum_i N_i)\rho}.
	\end{align*}
\end{proof}

	\begin{lemma}\label{lempb0}
	Let \(\omega\) and \(\sigma\) be quantum states, and let \(M,N\geq0\) satisfy
	\begin{align*}
		M+N\leq\mathbb I,\hspace{4mm}
		\frac{\operatorname{Tr}M\sigma}
		{\operatorname{Tr}(M+N)\sigma}
		\leq q,
	\end{align*}
	here \(q\in(0,1)\).	Let
	$
	p=\operatorname{Tr}(M+N)\omega,
	\epsilon=
	\frac{\operatorname{Tr}N\omega}
	{\operatorname{Tr}(M+N)\omega}.
	$	If \(p\geq c>0\) and \(\epsilon\leq 1/2\), then
	$$
	-\log\epsilon
	\leq
	D_{\max}(\sigma\|\omega)
	+
	D_H^{1-c/2}(\omega\|\sigma)
	+
	\log\frac{q}{1-q}.
	\label{S-NV3}
	$$
\end{lemma}
\begin{proof}
	As $\frac{\mathrm{ tr}M\sigma}{\mathrm{ tr}(M+N)\sigma}\le q,$
	
	$$
	\operatorname{Tr}M\sigma
	\leq
	q\,\operatorname{Tr}(M+N)\sigma,
	$$
then we have
	$$
	\operatorname{Tr}N\sigma
	\geq
	\frac{1-q}{q}\operatorname{Tr}M\sigma.
	\label{S-NV4}
	$$
	By the definition of the max-relative entropy,
	
	$$
	\sigma
	\leq
	2^{D_{\max}(\sigma\|\omega)}\omega.
	$$
	
	Since \(N\geq0\), this gives
	\begin{align*}
		\operatorname{Tr}N\omega
		\geq&
		2^{-D_{\max}(\sigma\|\omega)}
		\operatorname{Tr}N\sigma\\
		\geq&
		2^{-D_{\max}(\sigma\|\omega)}
		\frac{1-q}{q}
		\operatorname{Tr}M\sigma.
	\end{align*}
	As $
	\operatorname{Tr}M\omega
	=
	p(1-\epsilon)
	\geq
	\frac c2
	\label{S-NV7}
	$, \(M\) is feasible for
	\(D_H^{1-c/2}(\omega\|\sigma)\), hence,
	$$
	-\log\operatorname{Tr}M\sigma
	\leq
	D_H^{1-c/2}(\omega\|\sigma).
	\label{S-NV8}
	$$
	Finally, since \(p\leq1\), $
	\epsilon
	=
	\frac{\operatorname{Tr}N\omega}{p}
	\geq
	\operatorname{Tr}N\omega.
	$
	By the monotonicity of $-\log(\cdot)$, we arrive at
	$$
	-\log\epsilon
	\leq
	D_{\max}(\sigma\|\omega)
	+
	D_H^{1-c/2}(\omega\|\sigma)
	+
	\log\frac{q}{1-q},
	$$
	which proves the claim.
	
\end{proof}

\subsection{Entanglement}

Assume $\mathcal{H}_{AB}$ is the Hilbert space with finite dimensions. A state $\rho_{AB}$ is separable if it can be written as $$\rho=\sum_i p_i\rho_i^A\otimes\rho_i^B,$$ 
here the states $\rho^A_i$ and $\rho_i^B$ are states on local systems $A$ and $B,$ respectively. Otherwise, $\rho_{AB}$ is entangled. We will denote the set of separable states of $\mathcal{H}_{AB}$ as $Sep_{A:B}$, or simply $Sep$ if there is no ambiguity regarding the system.otherwise, it is entangled. Besides, we denote $\overline{Sep}$ and $cone(Sep)=\{\lambda\sigma|\lambda>0,\sigma\in Sep\}$ as the set of separable substates and cone of separable states.

An important method to detect whether a state is separable is the positive partial transpose(PPT) criterion \cite{peres}, which said any separable state $\rho_{AB}$ satisfies the following inequality $\rho_{AB}^{T_B}\ge 0$. A bipartite state $\sigma$ satisfying the PPT criterion is called a PPT state. Furthermore, we can generalize the above concepts to the POVMs. A measurement $\mathsf{M}$ is said to be separable measurements if 
\begin{align*}
	\mathsf{M}=\{M_x|\sum_x M_x=\mathbb{I}, M_x\in \overline{Sep}\}.
\end{align*}
Here we denote the set of all separable measurements as $\mathbb{SEP}.$ Besides, we denote $\mathbb{ALL}$ as the set of all measurements, $\mathbb{ALL}=\{(M_x)|\sum_xM_x=\mathbb{I},M_x\ge 0,\forall x\}$.

\subsection{Probabilistic entanglement distillation exponents for the resource theory of entanglement}\label{aa0}

In this section, we will present the analytical formula for the entanglement distillation exponents and entanglement cost exponents under $\mathbb{O}_{\mathcal{NE}}^{\delta}$ and $\mathbb{O}_{\mathcal{DNE}}^{\delta}$ with the approach of semidefinite programm(SDP) .

Assume $\mathcal{H}_{AB}$ is a bipartite system with $dim(\mathcal{H}_A)=dim(\mathcal{H}_B)=m\ge2$. Let $\ket{\psi_d}=\frac{1}{\sqrt{d}}\sum_i\ket{ii}$ be the maximally entangled state(MES) of $\mathcal{H}_{AB}$. An important property of the MES is that it stays unchanged under $\mathcal{T}(\cdot)$,
here $\mathcal{T}(\cdot)=\int_UdU (U\otimes \overline{U})^{\dagger}\cdot(U\otimes \overline{U})$ is local operations and shared randomness, hence, it can be realized by local operations and classical communication (LOCC).
Next based on the Schur-weyl theorem, $\mathcal{T}(X)$ can be written as follows,
\begin{align*}
	\mathcal{T}(X)=\Psi_m\mathrm{tr}(X\Psi_m)+\tau_m\mathrm{tr}[X(\mathbb{I}-\Psi_m)],
\end{align*}
here $\Psi_m=\ket{\psi_m}\bra{\psi_m}$, $\tau_m=\frac{\mathbb{I}-\Psi_m}{m^2-1}$. Assume $\mathcal{N}$ is a subchannel, then 
\begin{align}
	\mathcal{T}\circ\mathcal{N}(X)=&\Psi_m\mathrm{ tr}\mathcal{N}(X)\Psi_m+\tau_m\mathrm{ tr}\mathcal{N}(X)(\mathbb{I}-\Psi_m),\label{sd2}
\end{align}
 For (\ref{sd2}), 
\begin{align*}
	\Psi_m\mathrm{ tr}\mathcal{N}(X)\Psi_m+\tau_m\mathrm{ tr}\mathcal{N}(X)(\mathbb{I}-\Psi_m)=\Psi_m\mathrm{ tr}X\mathcal{N}^{\dagger}(\Psi_m)+\tau_m\mathrm{ tr}X\mathcal{N}^{\dagger}(\mathbb{I}-\Psi_m),\\
\end{align*}
here as $\mathcal{N}$ is a subchannel, then $\mathcal{N}^{\dagger}(\mathbb{I})\le \mathbb{I}$, hence, (\ref{sd2}) can always be written as the following,
\begin{align*}
	\Lambda_{M,N}(X)=\mathrm{ tr}MX\cdot \Psi_m+\mathrm{ tr}NX\cdot \tau_m,
\end{align*}
here $M,N\ge 0$ and $M+N\le \mathbb{I}.$ Next we present properties of the subchannels  $\Lambda_{M,N}(X)$  needed here.

\begin{lemma}\label{dm}
	Assume $\mathcal{H}_{AB}$ is a Hilbert space with $dim(\mathcal{H}_A)=dim(\mathcal{H}_B)=m,$ both $M$ and $N$ are semidefinite positive operators acting on $\mathcal{H}_{AB}$ with $M+N\le \mathbb{I}_{AB}$, let $\Lambda_{M,N}(X)=\mathrm{ tr}MX\cdot \Psi_m+\mathrm{ tr}NX\cdot \tau_m$, here $\tau_m=\frac{I-\Psi_m}{m^2-1}$, then for any $\varepsilon>0$, we have
	\begin{itemize}
		\item[(1).] $\Lambda_{M,N}(\cdot)\in \mathcal{NE}_{\varepsilon}$ if and only if $\sup_{X\in Sep}\frac{\mathrm{ tr}MX}{\mathrm{ tr}NX}\le \frac{\epsilon+1}{m-1-\epsilon}$.
		\item[(2).] $\Lambda_{M,N}(\cdot)\in \mathcal{DNE}_{\varepsilon}$ if and only if $\sup_{X\in Sep}\frac{\mathrm{ tr}MX}{\mathrm{ tr}NX}\le \frac{\epsilon+1}{m-1-\epsilon}$ and $N,\frac{1}{m}M+(1-\frac{1}{m})N\in \overline{Sep}.$
	\end{itemize}
\end{lemma}

\begin{proof}
	\begin{itemize}
		\item[(1).]  For a substate $\gamma=p\Psi_m+q\tau_m$, then 
		\begin{align*}
			\frac{\gamma}{\mathrm{ tr}\gamma}\in Sep\Longleftrightarrow p\in [0,\frac{q}{m-1}],\\
			R_G(\frac{\gamma}{\mathrm{ tr}\gamma})=\max\{0,\frac{mp}{p+q}-1\}.
		\end{align*} Hence, for any separable state $X$, $\frac{\mathrm{ tr}MX}{\mathrm{ tr}NX}\le \frac{\epsilon+1}{m-1-\epsilon}.$
		we finish the proof.
		\item[(2).] As $\Lambda_{M,N}(\cdot)$ is $\mathcal{DNE}_{\varepsilon}$ if and only if $\Lambda_{M,N}\in \mathcal{NE}_{\varepsilon}$ and $\Lambda^{\dagger}_{M,N}(Sep)\subset cone(Sep)$. Besides, $\Lambda_{M,N}^{\dagger}(X)=\mathrm{tr}\Psi_mX\cdot M+\mathrm{ tr}\tau_mX\cdot N,$ when $X$ is a separable state, $\mathrm{ tr}X\Psi_m\in (0,\frac{1}{m}).$ Hence, $\Lambda^{\dagger}_{M,N}(Sep)\subset\overline{Sep}$ if and only if $N\in \overline{Sep}$ and $\frac{1}{m}M+(1-\frac{1}{m})N\in \overline{Sep}.$
	\end{itemize}
\end{proof}
\begin{lemma}\label{lfd}
Assume $\rho_{AB}$ is a bipartite state, then its probabilistic distillation exponent for the maximally entangled state $\ket{\psi_m}$ under the $\mathbb{O}_{\mathcal{F}}^{\delta}$ instrument, $E_{d,err,p}^{(m),\mathbb{O}_{\mathcal{F}}^{\delta}}(\rho_{AB})$, can be rewritten as
\begin{align*}
	E_{d,err,p}^{(m),\mathbb{O}_{\mathcal{F}}^{\delta}}(\rho_{AB})=&\liminf_{n\rightarrow\infty}\sup-\frac{1}{n}\log\epsilon_n\\
	\textit{s. t.}\hspace{4mm}&  \frac{\mathrm{tr}M\rho_{AB}^{\otimes n}}{\mathrm{ tr}(M+N)\rho_{AB}^{\otimes n}}\ge 1-\epsilon_n,
	\mathcal{E}_i(X)=\mathrm{ tr}MX\cdot\Psi_m+\mathrm{ tr}NX\cdot\tau_m,\\
	&M+N\le \mathbb{I},M,N\ge 0,\\
	&\mathcal{E}_i\in \mathcal{E},\mathcal{E}\in \mathbb{O}_{\mathcal{NE}}^{\delta},\mathcal{F}=\{\mathcal{NE},\mathcal{DNE}\}.
\end{align*}
Here $\tau_m=\frac{\mathbb{I}-\Psi_m}{m^2-1}$.
\end{lemma}
\begin{proof}
As $\mathcal{T}(\cdot)=\int_U (U\otimes\overline{U})^{\dagger}(\cdot)(U\otimes\overline{U})\in \mathcal{F}$, when $\mathcal{E}_i(\cdot)\in \mathcal{E}$ and $\mathcal{E}\in\mathbb{O}_{\mathcal{F}}^{\delta},$ $\mathcal{T}\circ\mathcal{E}_i\in \mathcal{T}\circ\mathcal{E}$, $\mathcal{T}\circ\mathcal{E}\in\mathbb{O}_{\mathcal{F}}^{\delta}$.  then
	\begin{align*}
	F(\frac{\mathcal{E}_i(\rho_{AB}^{\otimes n})}{\mathrm{tr}(\mathcal{E}_i(\rho_{AB}^{\otimes n}))},\Psi_m)=&\bra{\psi_m}\frac{\mathcal{E}_i(\rho_{AB}^{\otimes n})}{\mathrm{ tr}\mathcal{E}_i(\rho_{AB}^{\otimes n})}\ket{\psi_m}\\
	=&\frac{\int_U\bra{\psi_m}(U\otimes \overline{U})^{\dagger}\mathcal{E}_i(\rho_{AB}^{\otimes n})(U\otimes\overline{U})\ket{\psi_m}dU}{\mathrm{ tr}\mathcal{E}_i(\rho_{AB}^{\otimes n})}\\
	=&\frac{\mathrm{ tr}\Psi_m(\mathrm{ tr}\mathcal{E}_i^{\dagger}(\Psi_m)\rho^{\otimes n}\cdot\Psi_m+\mathrm{ tr}\mathcal{E}_i^{\dagger}(\mathbb{I}-\Psi_m)\rho_{AB}^{\otimes n}\cdot\tau)}{\mathrm{ tr}\mathcal{E}_i(\rho_{AB}^{\otimes n})}\\
	=&\frac{\mathrm{ tr}M\rho_{AB}^{\otimes n}}{\mathrm{tr}(M+N)\rho_{AB}^{\otimes n}},
	\end{align*}
	Here $M=\mathcal{E}_i^{\dagger}(\Psi_m)$ and $N=\mathcal{E}_i^{\dagger}(\mathbb{I}-\Psi_m)$. As $\mathcal{E}_i$ is a subchannel, $M+N=\mathcal{E}_i(\mathbb{I})\le \mathbb{I}$.
	Hence, we finish the proof.
\end{proof}

\emph{Theorem \ref{t1}:} Assume $\rho_{AB}^{\otimes n}$ is a bipartite state on $\mathcal{H}_{AB}^{\otimes n}$. Let $m\ge2\in \mathbb{N}$, $\delta\in (0,1)$, then the conditional error exponent of the entanglement distillation under $\mathbb{O}_{\mathcal{NE}}^{\delta}$  is 
\begin{align*}
	E_{d,err,p}^{(m),\mathbb{O}_{\mathcal{NE}}^{\delta}}(\rho^{\otimes n}_{AB})=\frac{1}{n}\hat{\beta}_{\frac{\delta+1}{m},Sep}(\rho_{AB}^{\otimes n}).
\end{align*}

Besides, the zero-rate conditional error exponent of probabilistic entanglement distillation under $\mathbb{O}_{\mathcal{NE}}^{\delta}$ for $\rho_{AB}$ is 
\begin{align*}
	E_{d,err,p}^{\infty,\mathbb{O}_{\mathcal{NE}}^{\delta}}(\rho_{AB}):=\lim\limits_{n\rightarrow\infty}E_{d,err,p}^{(m_n),\mathbb{O}_{\mathcal{NE}}^{\delta}}(\rho^{\otimes n}_{AB})=\hat{D}_{\Omega,Sep}^{reg }(\rho).
\end{align*}
\begin{proof}
	Assume $\{\mathcal{E}_i\}_{i=1}^k$ is a feasible $\mathbb{O}_{\mathcal{NE}}^{\delta}$ instrument such that
	\begin{align}
		1-\epsilon_n\le& F(\frac{\mathcal{E}_i(\rho^{\otimes n})}{\mathrm{tr}(\mathcal{E}_i(\rho^{\otimes n}))},\Psi_m)\nonumber\\
		=&\bra{\psi_m}\frac{\mathcal{E}_i(\rho^{\otimes n})}{\mathrm{tr}(\mathcal{E}_i(\rho^{\otimes n}))}\ket{\psi_m}\nonumber\\
		=&\frac{\mathrm{tr}[\rho_{AB}^{\otimes n}\mathcal{E}_i^{\dagger}(\Psi_m)]}{p_n},\label{tf0}
	\end{align}
	here $\mathcal{E}_i^{\dagger}(\cdot)$ satisfies $\mathrm{tr}(\mathcal{E}_i^{\dagger}(A)B)=\mathrm{tr}(A\mathcal{E}_i(B))$, $p_n=\mathrm{tr}(\mathcal{E}_i^{\dagger}(\mathbb{I})\rho_{AB}^{\otimes n})$.
	
	Let $M^{(n)}_2=\mathcal{E}_i^{\dagger}(\Psi_m), M^{(n)}_1=\mathcal{E}^{\dagger}_i(\mathbb{I}-\Psi_m)$, $M^{(n)}_0=\mathbb{I}-M^{(n)}_1-M_2^{(n)}$. As $\mathcal{E}_i$ is completely positive and trace nonincreasing and $\mathbb{I}-\Psi_m\ge 0$, $M^{(n)}_1,M^{(n)}_2\ge 0$. As $\sum_{i=1}^k\mathcal{E}_i$ is trace preserving, then $\sum_i\mathcal{E}_i^{\dagger}(\mathbb{I})=\mathbb{I}$, 
	\begin{align*}
		M^{(n)}_0=&\mathbb{I}-M^{(n)}_1-M^{(n)}_2\\
		=&\sum_{\{1,2,\cdots,k\}-i}\mathcal{E}_l^{\dagger}(\mathbb{I})\ge 0,
	\end{align*} 
	Hence, $\{M^{(n)}_0,M^{(n)}_1,M^{(n)}_2\}$ is a POVM. Next based on (\ref{tf0}), we have
	\begin{align}
		\frac{\mathrm{tr}M_1^{(n)}\rho^{\otimes n}}{\mathrm{ tr}(M_1^{(n)}+M_2^{(n)})\rho^{\otimes n}}\le \epsilon_n.
	\end{align}
	
	Assume $\sigma_n$ is an arbitrary separable state in $\mathcal{H}_{AB}^{\otimes n}$, then
	\begin{align*}
		\frac{\mathrm{tr}M^{(n)}_2\sigma_n}{\mathrm{tr}(M^{(n)}_1+M^{(n)}_2)\sigma_n}=&\frac{\mathrm{tr}(\mathcal{E}_i(\sigma_n)\Psi_m)}{\mathrm{tr}\mathcal{E}_i(\sigma_n)}\\
		\le& \frac{\delta+1}{m},
	\end{align*}
	the last inequality is due to Lemma \ref{dm}.  Thus, based on the definition of postselected hypothesis testing, we have
	\begin{align*}
		\hat{\beta}_{\frac{\delta+1}{m}, Sep}(\rho_{AB}^{\otimes n})\ge& \log\frac{\mathrm{tr}(M^{(n)}_1+M^{(n)}_2)\rho^{\otimes n}}{\mathrm{tr}(M^{(n)}_1\rho^{\otimes n})}\\
		\ge&-\log\epsilon_n.
	\end{align*}
	Then multiplying two sides $\frac{1}{n}$, and when taking the supermum over all $\{\mathcal{E}_i\}_{i=1}^k\in \mathbb{O}_{\mathcal{NE}}^{\delta},$ we have
	\begin{align}
		E^{(m),\mathbb{O}_{\mathcal{NE}}^{\delta}}_{d,err,p}(\rho_{AB} )\le&	\frac{1}{n}\min_{\sigma_n\in Sep_{A_n:B_n}} \hat{\beta}_{\frac{\delta+1}{m}}(\rho_{AB}^{\otimes n},\sigma_n)\nonumber\\
		=&\frac{1}{n} \hat{\beta}_{\frac{\delta+1}{m},Sep}(\rho_{AB}^{\otimes n}).\label{a1}
	\end{align}
	Next we show the other direction. Assume $\{(M^{(n)}_1,M^{(n)}_2,M^{(n)}_0)|\sum_{i=0}^2M_i^{(n)}=\mathbb{I},M^{(n)}_i\ge0,i=1,2,0\}$ is a POVM with $\frac{\mathrm{tr}M^{(n)}_2\sigma_n}{\mathrm{ tr}M_1^{(n)}\sigma_n}\le \frac{\delta+1}{m-1-\delta}$ for any $\sigma_n\in Sep(A_n:B_n).$ Next let
	\begin{align*}
		\mathcal{E}_1(\cdot)=\mathrm{tr}(M^{(n)}_2(\cdot))\Psi_m+\mathrm{tr}M^{(n)}_1(\cdot)\frac{\mathbb{I}-\Psi_m}{m^2-1},\\
		\mathcal{E}_2(\cdot)=\mathrm{tr}M_0^{(n)}(\cdot)\frac{\mathbb{I}-\Psi_m}{m^2-1}.
	\end{align*}
	For any $\sigma_n\in Sep_{A_n:B_n}$, as $\mathbb{I}-\Psi_m$ is separable, $\mathcal{E}_2$ is $\mathcal{NE}$. Based on Lemma \ref{dm}, and  $\frac{\mathrm{tr}M^{(n)}_2\sigma_n}{\mathrm{ tr}M_1^{(n)}\sigma_n}\le \frac{\delta+1}{m-\delta-1}$, $\mathcal{E}_1(\cdot)$ is in $\mathcal{NE}_{\delta}.$ As $\mathrm{tr}(\mathcal{E}_1(\cdot)+\mathcal{E}_2(\cdot))=\mathrm{tr}(\cdot)$, and $\mathcal{E}_1$ and $\mathcal{E}_2$ is completely positive, $\{\mathcal{E}_1,\mathcal{E}_2\}$ is an $\mathbb{O}_{\mathcal{NE}}^{\delta}$ instrument. Then we have
	\begin{align*}
		E_{d,err,p}^{(m),\mathbb{O}_{\mathcal{NE}}^{\delta}}(\rho_{AB} )\ge& -\frac{1}{n}\log\frac{\mathrm{tr}M_1^{(n)}\rho^{\otimes n}}{\mathrm{tr}(M_1^{(n)}+M_2^{(n)})\rho^{\otimes n}}
	\end{align*}
	by optimising over all measurements with the property, we have
	\begin{align}
		E_{d,err,p}^{(m),\mathbb{O}_{\mathcal{NE}}^{\delta}}(\rho_{AB} )\ge& \frac{1}{n}\hat{\beta}_{\frac{\delta+1}{m},Sep}(\rho^{\otimes n}).\label{a2}
	\end{align}
	
	Based on (\ref{a1}) and (\ref{a2}), we have
	\begin{align}
		E_{d,err,p}^{(m),\mathbb{O}_{\mathcal{NE}}^{\delta}}(\rho )= \frac{1}{n}\hat{\beta}_{\frac{\delta+1}{m},Sep}(\rho^{\otimes n})\label{a3}
	\end{align}

	For for the probabilistic entanglement distillation under zero rate asymptotic scenario, by combing Corollary \ref{c1} and (\ref{a3}), let $n\rightarrow\infty,$
	\begin{align*}
		E^{\infty,\mathbb{O}_{\mathcal{NE}}^{\delta}}_{d,err,p}(\rho_{AB})=&D_{\Omega,Sep}^{reg}(\rho).
	\end{align*}

\end{proof}
 
\emph{Theorem \ref{th2}} Assume $\rho_{AB}^{\otimes n}$ is a bipartite state on $\mathcal{H}_{AB}^{\otimes n}$. Let $\delta\in (0,1)$ and $m\in \mathbb{N}\ge 2$, then the error exponent of probabilistic entanglement distillation under $\mathbb{O}_{\mathcal{DNE}}^{\delta}$ satisfies
\begin{align*}
	\hat{\beta}^{\mathbb{SEP}}_{\frac{\delta+2}{m+1},Sep}(\rho_{AB}^{\otimes n})- \log\frac{m+1}{m}\ge&	nE_{d,err,p}^{(m),\mathbb{O}_{\mathcal{DNE}}^{\delta}}(\rho_{AB}^{\otimes n})\ge \hat{\beta}_{\frac{\delta+1}{m},Sep}^{\mathbb{SEP}}(\rho_{AB}^{\otimes n}).
\end{align*}
Besides, the zero-rate conditional error exponent of probabilistic entanglement distillation under $\mathbb{O}_{\mathcal{NE}}^{\delta}$ for $\rho_{AB}$ is 
\begin{align*}
	\liminf\limits_{n\rightarrow\infty} E^{\infty,\mathbb{O}_{\mathcal{DNE}}^{\delta}}_{d,err,p}(\rho_{AB})=&\hat{D}_{\Omega,Sep}^{reg,\mathbb{SEP}}(\rho_{AB}).
\end{align*}
\begin{proof}
	Based on Lemma \ref{lfd}, we only need to consider the map of the form $\Lambda(X)=\mathrm{ tr}MX\cdot\Psi_m+\mathrm{ tr}NX\tau_m,$ here $M+N\le \mathbb{I},$ $M,N\ge0.$ Based on Lemma \ref{dm}, 
	$$\Lambda(X)\in \mathcal{DNE}_{\delta}\Longleftrightarrow\sup\limits_{X\in Sep}\frac{\mathrm{ tr}MX}{\mathrm{ tr}NX}\le\frac{\delta+1}{m-\delta-1}, N,\frac{1}{m}M+(1-\frac{1}{m})N\in cone(Sep).$$
	
	 The condition that the probabilistic entanglement distillation subchannel turns $\rho$ to $\Psi_m$ up to error $\epsilon$ probabilistically
	\begin{align*}
		\frac{\mathrm{ tr}M\rho_{AB}^{\otimes n}}{\mathrm{ tr}(M+N)\rho_{AB}^{\otimes n}}\ge 1-\epsilon,
	\end{align*}
	
	Next assume $(M,N,\mathbb{I}-M-N)$ is a feasible measurement of $\hat{\beta}_{\frac{\delta+1}{m-\delta-1},Sep}^{\mathbb{SEP}}(\cdot)$, as $M,N,\mathbb{I}-M-N\in cone(Sep)$, $N,\frac{1}{m}M+(1-\frac{1}{m})N\in cone(Sep)$, and $\sup_{\sigma\in Sep}\frac{\mathrm{tr}M\sigma}{\mathrm{ tr}N\sigma}\le\frac{\delta+1}{m-\delta-1}.$ Furthermore, let $\mathcal{E}_2(X)=\mathrm{ tr}(\mathbb{I}-M-N)X\cdot\tau_m$, then 
	\begin{align}
		E_{d,err,p}^{(m),\mathbb{O}_{\mathcal{DNE}}^{\delta}}(\rho^{\otimes n})\ge &-\frac{1}{n}\log\frac{\mathrm{ tr}N\rho_{AB}^{\otimes n}}{\mathrm{ tr}(M+N)\rho_{AB}^{\otimes n}}\nonumber\\
		\ge& \frac{1}{n} \hat{\beta}_{\frac{\delta+1}{m},Sep}^{\mathbb{SEP}}(\rho_{AB}^{\otimes n}).\label{th2f1}
	\end{align}
	The last inequality is due to the definition of $\hat{\beta}^{\mathbb{SEP}}_{\epsilon,{Sep}}(\rho_{AB}).$
	
	Next assume $\mathcal{E}=\{(\mathcal{E}_1,\mathcal{E}_2)|\mathcal{E}_1,\mathcal{E}_2\in \mathcal{DNE}_{\delta},\mathrm{tr}[\mathcal{E}_1(\cdot)+\mathcal{E}_2(\cdot)]=\mathrm{ tr}(\cdot)\}$ is a feasible one-shot distillation protocol such that 
	\begin{align*}
		F(\frac{\mathcal{E}_1(\rho^{\otimes n})}{\mathrm{ tr}\mathcal{E}_1(\rho^{\otimes n})},\Psi_m)\ge 1-\epsilon.
	\end{align*}
	Here we always assume $\mathcal{E}_1(X)=\mathrm{ tr}MX\cdot\Psi_m+\mathrm{ tr}NX\cdot\tau_m$, $M+N\le \mathbb{I},$ $M,N\ge 0.$ As $\mathcal{E}_1\in\mathcal{DNE}_{\delta}$, $\sup_{X\in Sep}\frac{\mathrm{ tr}MX}{\mathrm{ tr}NX}\le \frac{\delta+1}{m-\delta-1}$ and $N,M+(m-1)N\in cone(Sep)$. Then let $M_2=\mathcal{E}_1^{\dagger}(\frac{1}{m+1}(\mathbb{I}+m\Psi_m)),M_1=\mathcal{E}_1^{\dagger}(\frac{m}{m+1}(\mathbb{I}-\Psi_m)),$ for any separable state $\sigma$,
	\begin{align*}
\frac{\mathrm{ tr}M_2\sigma}{\mathrm{ tr}(M_1+M_2)\sigma}=&\frac{\mathrm{ tr}(\frac{1}{m+1}(\mathbb{I}+m\Psi_m))\mathcal{E}_1(\sigma)}{\mathrm{tr}(\mathcal{E}_1(\sigma))}\\
\le&\frac{\delta+2}{m+1},
	\end{align*}
	here the last inequality is due to the definition of $\mathcal{DNE}_{\delta}$ and Lemma \ref{dm},
	then 
	\begin{align}
	\hat{\beta}^{\mathbb{SEP}}_{\frac{\delta+2}{m+1},Sep}(\rho_{AB}^{\otimes n})\ge&\frac{1}{n}\log \frac{\mathrm{ tr}(M_1+M_2)\rho_{AB}^{\otimes n}}{\mathrm{ tr}M_1\rho_{AB}^{\otimes n}}\nonumber\\
		=&\frac{1}{n}\log\frac{1}{\mathrm{ tr}\frac{m}{m+1}(\mathbb{I}-\Psi_{m})\frac{\mathcal{E}_1(\rho_{AB}^{\otimes n})}{\mathrm{tr}\mathcal{E}_1(\rho^{\otimes n})}}\nonumber\\
		\ge&\frac{1}{n}\log\frac{1}{\epsilon}+\frac{1}{n}\log\frac{m+1}{m}\nonumber
	\end{align}
	Here the first inequality is due to the definition of $\hat{\beta}^{\mathbb{SEP}}_{\epsilon,Sep}(\rho)$.
	
	At last, for the probabilistic entanglement distillation under zero rate asymptotic scenario, based on (\ref{th2f1}) and (\ref{th2f2}), we have
	\begin{align}
		\lim\limits_{n\rightarrow\infty}\frac{1}{n}\hat{\beta}_{\frac{\delta+1}{m_n},Sep}^{\mathbb{SEP}}(\rho_{AB}^{\otimes n})\le\lim\limits_{n\rightarrow\infty} E_{d,err,p}^{(m_n),\mathbb{O}_{\mathcal{DNE}}^{\delta}}(\rho_{AB})\label{th2f3}\\
\lim_{n\rightarrow\infty}\frac{1}{n}	\hat{\beta}^{\mathbb{SEP}}_{\frac{\delta+2}{m_n+1},Sep}(\rho^{\otimes n})\ge \lim\limits_{n\rightarrow\infty}E_{d,err,p}^{(m_n),\mathbb{O}_{\mathcal{DNE}}^{\delta}}(\rho_{AB})\label{th2f4},
	\end{align}
	Hence, combing $(\ref{th2f3})$, (\ref{th2f4}) and Corollary \ref{c2}, we have
	\begin{align*}
		\hat{D}_{\Omega,Sep}^{reg,\mathbb{SEP}}(\rho)=E_{d,err,p}^{\infty,\mathbb{O}_{\mathcal{DNE}}^{\delta}}(\rho_{AB}).
	\end{align*}
\end{proof}

\subsection{Analytical Formula for Specific Examples} \label{aa1}
\emph{Example \ref{e1}:} Assume $\mathcal{H}_{AB}$ is a bipartite system with $dim(\mathcal{H}_A)=dim(\mathcal{H}_B)=d$, and $\rho_{AB}$ is the Werner state,
\begin{align*}
	\rho_{p}=p\cdot\frac{2P_s}{d(d+1)}+(1-p)\cdot\frac{2P_{as}}{d(d-1)},
\end{align*}
here $P_s=\frac{I+F}{2}$, $P_{as}=\frac{I-F}{2}$, $F$ is the swap operator, $F=\sum_{ij}\ket{ij}\bra{ji}$. 
\begin{itemize}
	\item[(1)] Then for each $n\in\mathbb{N}$,
	\begin{equation*}
		\frac{1}{n}D_{\Omega}(\rho_p^{\otimes n})= \begin{cases} 
			\log\frac{1-p}{p}\hspace{3mm}p<\frac{1}{2} \\
			0\hspace{12mm}p\ge\frac{1}{2}
		\end{cases}  .
	\end{equation*}
	\item[(2)] $$\min\limits_{\sigma\in Sep_{A:B}}D(\sigma||\rho_p)= \begin{cases} 
		-\frac{1}{2}\log[4p(1-p)]\hspace{3mm}p<\frac{1}{2} \\
		0\hspace{27mm}p\ge\frac{1}{2}
	\end{cases}  .$$
	\item[(3)] 	\begin{align*}
		\hat{D}_{Sep}^{\mathbb{SEP}}(\rho_p)=\begin{cases} 
			\log\frac{d+1-2p}{2dp}\hspace{3mm}0<p<\frac{1}{2} \\
			0\hspace{20mm}p\ge\frac{1}{2}
		\end{cases}  .
	\end{align*}
\end{itemize}
	\begin{proof}
		\begin{itemize}
			\item[(1)] Due to the faithfulness of $D_{\Omega}(\cdot)$ in Lemma \ref{l1}, and when $p\ge \frac{1}{2}$, $\rho_{AB}$ is separable \cite{horodecki2009quantum}, we only need to consider the case when $p<\frac{1}{2}$.
	
	As $p_{\frac{1}{2}}$ is separable,
	\begin{align*}
		\frac{1}{n}D_{\Omega,Sep}(\rho_{p}^{\otimes n})\le& \frac{1}{n}D_{\Omega}(\rho_{p}^{\otimes n},\rho_{\frac{1}{2}}^{\otimes n})\\
		=&D_{\Omega}(\rho_p,\rho_{\frac{1}{2}}),
	\end{align*} 
	As $P_s$ and $P_{as}$ are two mutually orthogonal projectors, then 
	\begin{align*}
		D_{\Omega}(\rho_p,\rho_{\frac{1}{2}})=&\hspace{2mm}\log\inf\hspace{4mm} \mu\\
		\textit{s. t.}\hspace{3mm}&(1,1)\le(2p\lambda,2(1-p)\lambda)\le(\mu,\mu),\\
		&\lambda,\mu\ge 0.
	\end{align*}
	From computation, we have 
	\begin{align}
		D_{\Omega,Sep}(\rho_p)\le	D_{\Omega}(\rho_p,\rho_{\frac{1}{2}})=\log\frac{1-p}{p}. \label{lf1}\\
		\frac{1}{n}D_{\Omega,Sep}(\rho_p^{\otimes n})\le		\frac{1}{n}D_{\Omega}(\rho^{\otimes n}_p,\rho^{\otimes n}_{\frac{1}{2}})=\log\frac{1-p}{p}.\label{lf2}
	\end{align}Here $n$ is an arbitrary natural number.
	
 Next we show the other direction. 	The dual problem of $\Omega_{Sep}(\cdot)$ is
	\begin{align}
		\Omega_{Sep}(\rho)=\hspace{2mm}&\hspace{2mm}\sup\mathrm{tr}(A\rho)\label{dsep}\\
		\textit{s. t.}\hspace{4mm}&\mathrm{tr}B\rho=1,\nonumber\\
		&\mathrm{tr}(B-A)\sigma\ge 0\hspace{3mm}\forall \sigma\in Sep,\nonumber\\
		&A,B\ge 0.\nonumber
	\end{align} 
	Let
	\begin{align*}
		A=\frac{1}{p}P_{as},\hspace{3mm}B=\frac{1}{p}P_{s}.
	\end{align*}
	Due to computation, $\mathrm{tr}A\rho=\frac{1-p}{p},$ $\mathrm{tr}B\rho=1$. Next we show the last condition, when $\sigma$ is any separable state,
	\begin{align*}
		\mathrm{tr}(B-A)\sigma=&\frac{1}{p}\mathrm{tr}\frac{I+F-I+F}{2}\sigma\\\
		=&		\frac{1}{p}\mathrm{ tr}F\sigma\ge 0.
	\end{align*}
	Hence, $A$ and $B$ are feasible for the dual program of $\Omega_{Sep}(\cdot)$, then
	\begin{align}
		D_{\Omega,{Sep}}(\rho)\ge \log\frac{1-p}{p}. \label{rf1}
	\end{align}
	Combing $(\ref{lf1})$ and (\ref{rf2}), we have $$D_{\Omega,Sep}(\rho_p)=\log\frac{1-p}{p}.$$ For the state $\rho_p^{\otimes n}$, let 
	\begin{align*}
		A^{(n)}=A^{\otimes n},B^{(n)}=B^{\otimes n},
	\end{align*}
	Due to the computation, $\mathrm{tr}B^{(n)}\rho_p^{\otimes n}=1,$ for any separable state $\sigma_n\in Sep_{A_n:B_n}$,
	\begin{align*}
		&\mathrm{tr}(B^{(n)}-A^{(n)})\sigma_n\\
		&=\mathrm{tr}(\Pi_1+\Pi_3+\cdots+\Pi_{2\lceil\frac{n}{2}\rceil-1})\sigma_n,
	\end{align*}
	here $\Pi_m$ is a sum of all product opertors with $m$ $F$ and $n-m$ $I$. As each $\Pi_m$ satisfies $\mathrm{tr}\Pi_m\sigma\ge 0$, the above formula is nonnegative. Hence $A^{(n)}$ and $B^{(n)}$ are the feasible for $\rho_p^{\otimes n}$ in terms of $(\ref{dsep})$ for $\Omega_{Sep}$, then
	\begin{align}
		\frac{1}{n}	D_{\Omega,Sep}(\rho^{\otimes n})\ge \log\frac{1-p}{p},\label{rf2}
	\end{align}
	Combing (\ref{lf2}) and (\ref{rf2}), we have $$\frac{1}{n}D_{\Omega,Sep}(\rho_p)=\log\frac{1-p}{p}.$$
	
	\item[(2)]	Due to that quantum relative entropy satisfies the data-processing property, we have when $p\in [0,\frac{1}{2}),$
	\begin{align}
		\min_{\sigma\in Sep}D(\sigma||\rho_p)=\min_{q\in [\frac{1}{2},1]}[q\log\frac{q}{p}+(1-q)\log\frac{1-q}{1-p}],
	\end{align}
	Let $t(q)=q\log\frac{q}{p}+(1-q)\log\frac{1-q}{1-p}$, $t^{'}(q)=\log\frac{q}{1-q}-\log\frac{p}{1-p}$. As $r(q)=\frac{q}{1-q}$ is an increasing function, $$D(Sep||\rho_p)=D(\rho_{\frac{1}{2}}||\rho_p)=-\frac{1}{2}\log[4p(1-p)].$$
	
	At last, 
	\begin{align*}
		-\frac{1}{2}\log[4p(1-p)]-\log\frac{1-p}{p}=\frac{1}{2}\log\frac{1}{1-p}\cdot\frac{p}{(2(1-p))^2}\le 0.
	\end{align*}
	Hence, $$D_{\Omega}(\rho_p)\ge \min_{\sigma\in Sep}D(\sigma||\rho_p).$$
	\item[(3)]  Based on Corollary \ref{c2}, 
	\begin{align}
		\Omega_{Sep}^{\mathbb{SEP}}(\rho_p)=&\inf \gamma\label{ddsep}\\
		\textit{s. t.}\hspace{3mm}&\rho_p\preceq_{\mathbb{SEP}^{*}}t{\sigma}\preceq_{\mathbb{SEP}^{*}}\gamma\rho_p,\nonumber\\
		&	{\sigma}\in Sep, t>0.\nonumber
	\end{align}
	
	As $\rho_p$ satisfies $(U\otimes{U})^{\dagger} \rho_p(U\otimes{U})=\rho_p$, and $\int_UdU(U\otimes{U})^{\dagger} {\sigma}(U\otimes {U})=a \frac{2P_s}{d(d+1)}+(1-a)\frac{2P_{as}}{d(d-1)},$ as ${\sigma}\ge 0$, $a\ge \frac{ 1}{2}.$ As for any Werner state $W=m\frac{2P_s}{d(d+1)}+n\frac{2P_{as}}{d(d-1)}$, 
	\begin{align*}
		W\succeq_{\mathbb{SEP}^{*}}0\Longleftrightarrow \mathrm{ tr}W(\ket{xy}\bra{xy})=\frac{m(1+|\bra{x}y\rangle|^2)}{d(d+1)}+\frac{n(1-|\bra{x}y\rangle|^2}{d(d-1)})\ge0\Longleftrightarrow m\ge 0, (d-1)m+(d+1)n\ge 0,
	\end{align*}
	
	Hence, (\ref{ddsep}) turns into the following,
	\begin{align*}
		\Omega_{Sep}^{\mathbb{SEP}}(\rho_p)=&\inf\gamma\\
		\textit{s. t.}\hspace{3mm}& 	ta-p\ge 0, (d-1)(ta-p)+(d+1)(t(1-a)-1+p)\ge 0,\\
		&	\gamma p-ta\ge 0, (d-1)(\gamma p-ta)+(d+1)(\gamma-\gamma p-t(1-a))\ge 0, t>0.
	\end{align*}
	As when $p\ge\frac{1}{2},$ $\rho_p$ is separable, we only consider $0<p<\frac{1}{2}.$ Due to the computation, the optimal $\sigma\in Sep$ in (\ref{ddsep}) is $\rho_{\frac{1}{2}},$ we have
	\begin{align*}
		\hat{D}_{Sep}^{\mathbb{SEP}}(\rho_p)=\begin{cases} 
			\log\frac{d+1-2p}{2dp}\hspace{3mm}0<p<\frac{1}{2} \\
			0\hspace{20mm}p\ge\frac{1}{2}
		\end{cases}  .
	\end{align*}
	At last, we present a lower bound of  $\hat{D}_{Sep}^{\mathbb{SEP}}(\rho^{\otimes n}_p)$ when $n\in \mathbb{N}$ by presenting a feasible solution $(A_n,B_n)$ for $\rho_p^{\otimes n}$ in terms of the dual program for $\hat{D}_{Sep}^{\mathbb{SEP}}(\rho^{\otimes n}_p)$ in terms of $(\ref{dmhp})$.
	
	Let $B=\frac{1}{p}P_s,$ and $A(m,n)=mP_s+\frac{n(d+1)P_{as}}{d-1}$ with $n\le m$ and $m+\frac{d+1}{d-1}n\le \frac{1}{p}$, we have $A,B\in \overline{Sep}$ and $B-A\in Sep^{*}.$ Hence,
	\begin{align*}
		\hat{D}_{Sep}^{\mathbb{SEP}}(\rho_p)\ge& \log\mathrm{ tr}A\rho_p
		= \log(\frac{d+1-2p}{2dp}).
	\end{align*}
	Here $A=A(\frac{d-1}{2dp},\frac{d-1}{2dp})$,
	let $A^{(n)}=A^{\otimes n}$ and $B^{(n)}=B^{\otimes n}$, then $A^{(n)},B^{(n)}\in \overline{Sep}$. Assume $\sigma^{(n)}$ is any state in $Sep_{A^{\otimes n}:B^{\otimes n}}$,
	\begin{align*}
		\mathrm{ tr}(B^{(n)}-A^{(n)})\sigma^{(n)}=\sum_{k=1}^n\mathrm{tr}(A^{k-1}\otimes (B-A)\otimes B^{\otimes (n-k)})\sigma^{(n)}\ge 0.
	\end{align*}
	At last, we have $\frac{1}{n}D_{\Omega,Sep}^{\mathbb{SEP}}(\rho^{\otimes n}_p)\ge \log\frac{d+1-2p}{2dp}.$
	\end{itemize}

\end{proof}

\emph{Example \ref{e2}} Assume $\rho$ is a two-qubit state with 
\begin{align*}
	\rho=p_1\Psi^{+}+p_2\Psi^{-}+p_3\Phi^{+}+p_4\Phi^{-},\\
	\Psi^{+}=\frac{1}{2}(\ket{00}+\ket{11})(\bra{00}+\bra{11}),\\
	\Psi^{-}=\frac{1}{2}(\ket{00}-\ket{11})(\bra{00}-\bra{11}),\\
	\Phi^{+}=\frac{1}{2}(\ket{01}+\ket{10})(\bra{01}+\bra{10}),\\
	\Phi^{-}=\frac{1}{2}(\ket{01}-\ket{10})(\bra{01}-\bra{10}),
\end{align*}
where $p_i\ge 0$ and $\sum_i p_i=1.$ Then for each $n\in \mathbb{N}$,
Then for each $n\in \mathbb{N}$,
\begin{align*}
	E_{d,err,p}(\rho)=	  \begin{cases} 
		\log\frac{p_{max}}{1-p_{max}}\hspace{3mm}p_{max}\ge\frac{1}{2} \\
		0\hspace{18mm}p_{max}\le\frac{1}{2}
	\end{cases}.
\end{align*}
Here $p_{max}=\max\{p_1,p_2,p_3,p_4\}$. 
Furthermore,
\begin{align*}
	\min_{\sigma\in Sep}D(\sigma||\rho)=\begin{cases} 
		-\frac{\log[4p_{max}(1-p_{max})]}{2}\hspace{3mm}p_{max}\ge \frac{1}{2} \\
		0\hspace{30mm}p_{max}<\frac{1}{2}
	\end{cases} .
\end{align*}

\begin{proof}
		Due to the faithfulness of $D_{\Omega}(\cdot)$, and when $p_{max}\le \frac{1}{2}$, $\rho_{AB}$ is separable \cite{peres}. When $p_{max}>\frac{1}{2}$,
$$
\frac{1}{n}D_{\Omega,\mathrm{Sep}}(\rho^{\otimes n})\le \frac{1}{n}D_{\Omega}{(\rho^{\otimes n},\rho^{\otimes n}_{\frac{1}{2}})} = D_{\Omega}\big(\rho,\rho_{\frac{1}{2}}\big)
$$
	As $\Psi^{+}$, $\Psi^{-},\Phi^{+}$ and $\Phi^{-}$ are mutually orthogonal projectors, then
\begin{align*}
	D_{\Omega}(\rho,\rho_{\frac{1}{2}})&=\log\inf \mu\\
	\textit{s. t.} \hspace{3mm}& (p_1,p_2,p_3,p_4)\le \lambda(q_1,q_2,q_3,q_4)\le \mu (p_1,p_2,p_3,p_4),\\
	&\lambda,\mu\ge 0.
\end{align*} 
Here $q_i$ are the coefficients of $\rho_{\frac{1}{2}}$ and $\max_i q_i\le \frac{1}{2}$. Through computation, $D_{\Omega,Sep}(\rho)\le D_{\Omega}(\rho,\rho_{\frac{1}{2}})=\log\frac{p_{\max}}{1-p_{\max}}$. Based on the property of $D_{\Omega}(\cdot,\cdot)$, we have 
\begin{align}
	\frac{1}{n}D_{\Omega,Sep}(\rho^{\otimes n})\le \log\frac{p_{\max}}{1-p_{\max}}. \label{bell1}
\end{align}

Next we show the other direction. The dual problem of $\Omega_{Sep}(\cdot)$ is
\begin{align}
	\Omega_{Sep}(\rho)=&\hspace{2mm}\sup\mathrm{tr}(A\rho)\\
	\textit{s. t.}\hspace{4mm}&\mathrm{tr}B\rho=1,\nonumber\\
	&\mathrm{tr}(B-A)\sigma\ge 0\hspace{3mm}\forall \sigma\in Sep,\nonumber\\
	&A,B\ge 0.\nonumber
\end{align} 
Let
\begin{align*}
	A=\frac{1}{1-p_{max}}\Pi_k,\hspace{3mm}B=\frac{1}{1-p_{max}}(I-\Pi_k),
\end{align*}
here $\Pi_k$ is the projector of Bell state with the largest weight. As $\mathrm{tr}B\rho=1,$ for any separable state $\sigma,$
\begin{align*}
	\mathrm{ tr}(B-A)\sigma=&\frac{1}{1-p_{max}}\mathrm{tr}[(I-2\Pi_k)\sigma]\\
	=&\frac{1-2\mathrm{ tr}\Pi_k\sigma}{1-p_{max}}\ge 0,
\end{align*}
the last inequality is due to that $\mathrm{tr}\Pi_k\sigma\le \frac{1}{2}$, $\forall \sigma\in Sep.$ Hence,
\begin{align*}
	D_{\Omega,Sep}(\rho)\ge \frac{p_{max}}{1-p_{max}}.
\end{align*}

For the state $\rho^{\otimes n}$, let
\begin{align*}
	A^{(n)}=A^{\otimes n},\hspace{3mm} B^{(n)}=B^{\otimes n}.
\end{align*}

Through computation, $\mathrm{tr}B^{(n)}\rho^{\otimes n}=\mathrm{tr}^nB\rho=1.$ For any separable state $\sigma,$
\begin{align*}
	\mathrm{tr}(B^{(n)}-A^{(n)})\sigma=&\frac{1}{(1-p)^n}\mathrm{tr}[(I-\Pi_k)^{\otimes n}-\Pi_k^{\otimes n}]\sigma\\
	=&\frac{1}{(1-p)^n}\mathrm{tr}[(I-\Pi^{T_B}_k)^{\otimes n}-(\Pi^{T_B}_k)^{\otimes n}]\sigma^{T_B}\\
	=&\frac{1}{(1-p)^n}\mathrm{tr}[(I-\frac{F}{2})^{\otimes n}-(\frac{F}{2})^{\otimes n}]\sigma^{T_B}\ge 0,
\end{align*}
the last inequality is due to that $(I-\frac{F}{2})^{\otimes n}-(\frac{F}{2})^{\otimes n}\ge 0$ and $\sigma^{T_B}\in Sep$. Hence, $$D_{\Omega,Sep}(\rho^{\otimes n})\ge\log\mathrm{tr}A^{(n)}\rho^{\otimes n}=\log( \frac{p_{max}}{1-p_{max}})^n.$$
Combing (\ref{bell1}), we have
\begin{align*}
	\frac{1}{n}D_{\Omega,Sep}(\rho^{\otimes n})=\log\frac{p_{max}}{1-p_{max}}.
\end{align*}

Due to that quantum relative entropy satisfies the data-processing property, we have 
\begin{align}
	\min_{\sigma\in Sep}D(\sigma||\rho)=\min_{q_1\ge \frac{1}{2},\sum_{i=1}^4q_i=1}[q_1\log\frac{q_1}{p_1}+q_2\log\frac{q_2}{p_2}+q_3\log\frac{q_3}{p_3}+q_4\log\frac{q_4}{p_4}],\label{bellrr}
\end{align}
where we can always assume $q_1\ge q_2\ge q_3\ge q_4,$ $\max(p_1,p_2,p_3,p_4)=p_1$ and $p_i$ are the weights of Bell projectors corresponding to $q_i$ of $\sigma$ for the state $\rho$. By using Lagrange multiplier method, let
\begin{align*}
	f=q_1\log\frac{q_1}{p_1}+q_2\log\frac{q_2}{p_2}+q_3\log\frac{q_3}{p_3}+q_4\log\frac{q_4}{p_4}+\lambda(1-\sum_{i=1}^{4}q_i)+\mu(\frac{1}{2}-q_1),
\end{align*}
then KKT conditions are 
\begin{align*}
	\log\frac{q_1}{p_1}+1-\lambda-\mu=0,\\
	\log\frac{q_2}{p_2}+1-\lambda=0,\\
	\log\frac{q_3}{p_3}+1-\lambda=0,\\
	\log\frac{q_4}{p_4}+1-\lambda=0,\\
	q_1\ge \frac{1}{2},\hspace{3mm} \mu\ge 0,\hspace{3mm}\mu(\frac{1}{2}-q_1)=0.
\end{align*}
then \begin{align*}
	q_1=\frac{1}{2},\hspace{3mm} q_2=\frac{p_2}{2(1-p_1)}, \hspace{3mm}q_3=\frac{p_3}{2(1-p_1)},\hspace{3mm} q_4=\frac{p_4}{2(1-p_1)}.
\end{align*}
Hence, 
\begin{align*}
	\min_{\sigma\in Sep}D(\sigma||\rho)=\begin{cases} 
		-\frac{\log[4p_{max}(1-p_{max})]}{2}\hspace{3mm}p\ge \frac{1}{2} \\
		0\hspace{30mm}p<\frac{1}{2}
	\end{cases} .
\end{align*}
\end{proof}
 
 \emph{Theorem \ref{thepb0}} Assume $\rho_{AB}$ is a bipartite state, and $\delta\in (0,1)$. For any sequence of probabilistic entanglement distillation protocols under $\mathbb{O}_{\mathcal{NE}}^{\delta}$ with success probabilities $\liminf\limits_{n\rightarrow\infty}p_n>0,$ then 
	\begin{align*}
	E_{d,err,p>0}^{\infty,\mathbb{O}_{\mathcal{NE}}^{\delta}}(\rho_{AB})		\leq
			\inf_{\sigma\in Sep}D_{\max}(\sigma||\rho)+\inf_{\sigma\in Sep}D(\rho||\sigma)	.
	\end{align*} 
	
	\begin{proof}
	 Let $\sigma_n$ be the optimal separable state in $\mathcal{D}(\mathcal{H}_{A_nB_n})$, as $\delta\in (0,1)$, $q_n=\frac{1+\delta}{m_n}\in (0,1)$,	based on Lemma \ref{lempb0},
		\begin{align*}
			-\log \epsilon_n\le D_{\max}(\sigma_n||\rho^{\otimes n})+D_H^{1-\frac{c}{2}}(\rho^{\otimes n}||\sigma_n)+\log\frac{q_n}{1-q_n},
		\end{align*} 
		then
		\begin{align*}
			E_{d,err,p>0}^{\infty,\mathbb{O}_{\mathcal{NE}}^{\delta}}(\rho_{AB})\le& \inf_{\sigma_n\in Sep}\liminf\limits_{n\rightarrow\infty}\frac{1}{n}[D_{\max}(\sigma_n||\rho^{\otimes n})+D_{H}^{1-\frac{c}{2}}(\rho^{\otimes n}||\sigma_n)+\frac{q_n}{1-q_n}]\\
				\le&\inf_{\sigma\in Sep}\lim\limits_{n\rightarrow\infty}[\frac{1}{n}D_{\max}(\sigma^{\otimes n}||\rho^{\otimes n})+\frac{1}{n}D_H^{1-\frac{c}{2}}(\rho^{\otimes n}||\sigma^{\otimes n})+\frac{q_n}{n(1-q_n)}]\\
			=&\inf_{\sigma\in Sep}[D_{\max}(\sigma||\rho)+D(\rho||\sigma)]\\
		=&\inf_{\sigma\in Sep}D_{\max}(\sigma||\rho)+\inf_{\sigma\in Sep}D(\rho||\sigma)	.
		\end{align*}
		the last equality is due to that $\lim\limits_{n\rightarrow\infty}\frac{1}{n}\log\frac{q_n}{1-q_n}=\lim\limits_{n\rightarrow\infty}\frac{1}{n}\log\frac{1+\delta}{m_n-1-\delta}=0,$ $\lim\limits_{n\rightarrow\infty}\frac{1}{n}D_{\max}(\sigma^{\otimes n}||\rho^{\otimes n})=D_{\max}(\sigma||\rho)$ and $\lim\limits_{n\rightarrow\infty}\frac{1}{n}D_{H}^{1-\frac{c}{2}}(\rho^{\otimes n}||\sigma^{\otimes n})=D(\rho||\sigma).$
	\end{proof}

\subsection{Probabilistic entanglement distillation exponents for the resource theory of coherence}\label{apc}
\indent In this subsection, assume the reference basis is $\{\ket{i}\}$, let the permutation twirling be 
\begin{align*}
	\mathcal{T}_{S}(\rho)=\frac{1}{d!}\sum_{\pi\in S_m}P_{\pi}\rho P_{\pi},
\end{align*}
where $\pi$ takes over all the permutation in $S_{m}$, here $m\ge 2.$

\begin{lemma}Assume $\rho^{\otimes n}$ is a state, then its probabilistic distillation exponent for the maximally coherent state $\ket{\psi_m}$ under $\mathbb{O}_{\mathcal{NG}}^{\delta}$ instrument, $E_{d,err,p}^{(m),\mathbb{O}_{\mathcal{NG}}^{\delta}}(\rho^{\otimes n})$, can be rewritten as
	\begin{align*}
		E_{d,err,p}^{(m),\mathbb{O}^{\delta}_{\mathcal{NG}}}(\rho^{\otimes n})=&\sup-\frac{1}{n}\log\epsilon_n\\
		\textit{s. t.}\hspace{4mm}&  \frac{\mathrm{tr}M\rho^{\otimes n}}{\mathrm{ tr}(M+N)\rho^{\otimes n}}\ge 1-\epsilon_n,
		\mathcal{E}_i(X)=\mathrm{ tr}MX\cdot\Psi_m+\mathrm{ tr}NX\cdot\tau_m,\\
		&M+N\le \mathbb{I},M,N\ge 0,\\
		&\mathcal{E}_i\in \mathcal{E},\mathcal{E}\in\mathbb{O}^{\delta}_{\mathcal{NG}}.
	\end{align*}
\end{lemma}
\begin{proof}
	Based on Schur's Lemma,
	\begin{align*}
		\mathcal{T}_S(\rho)=a\Psi_m+b\tau_m,
	\end{align*}
	here $a=\mathrm{ tr}\Psi_m\rho$, $\tau_m=\frac{\mathbb{I}-\Psi_m}{m-1}$ and $b={1-a}.$
	Let $\{\mathcal{E}_i\}_i$ be a feasible instrument such that 
	\begin{align*}
	1-\epsilon_n\le& F(\frac{\mathcal{E}_i(\rho^{\otimes n})}{\mathrm{tr}(\mathcal{E}_i(\rho^{\otimes n}))},\Psi_m)\nonumber\\
	=&\frac{1}{d!}\sum_{\pi\in S_m}\mathrm{ tr}\frac{\mathcal{E}_i(\rho^{\otimes n})}{\mathrm{tr}(\mathcal{E}_i(\rho^{\otimes n}))}P_{\pi}\Psi_m P_{\pi}\nonumber\\
	=&\frac{\mathrm{ tr}\Psi_m [(\mathrm{ tr}\mathcal{E}_i(\rho^{\otimes n})\Psi_m)\Psi_m +(\mathrm{ tr}(\mathbb{I}-\Psi_m)\mathcal{E}_i(\rho^{\otimes n}))\tau_m]}{\mathrm{ tr}\mathcal{E}_i(\rho^{\otimes n})}\\
	=&\frac{\mathrm{ tr}\mathcal{E}_i(\rho^{\otimes n})\Psi_m}{\mathrm{ tr}\mathcal{E}_i(\rho^{\otimes n})}\\
	=&\frac{\mathrm{ tr}M\rho^{\otimes n}}{\mathrm{ tr}(M+N)\rho^{\otimes n}},
\end{align*}
here $M=\mathcal{E}_i^{\dagger}(\Psi_m),$ and $N=\mathcal{E}_i^{\dagger}(\mathbb{I}-\Psi_m)$. As $\mathcal{E}_i$ is trace-nonincreaing, $M+N=\mathcal{E}^{\dagger}_i(\mathbb{I})\le \mathbb{I}$. Hence, we finish the proof.
\end{proof}
\begin{lemma}\label{cm}\cite{napoli2016robustness,piani2016robustness}
	Assume $\gamma=p\ket{\psi_m}\bra{\psi_m}+q\tau_m$, here $p+q\le 1$, $p,q\ge0,$ and $\tau_m=\frac{\mathbb{I}-\ket{\psi_m}\bra{\psi_m}}{m-1}$, then
	\begin{align*}
		R_{g,ic}(\frac{\gamma}{\mathrm{ tr}\gamma})=\begin{cases}
			m\frac{p}{p+q}-1 \hspace{7mm}\frac{p}{p+q}\ge \frac{1}{m}\\
			\frac{1-\frac{mp}{p+q}}{m-1}\hspace{12mm}\frac{p}{p+q}\le \frac{1}{m}.
		\end{cases}
	\end{align*}
\end{lemma}

\emph{Theorem \ref{th3}}:
Assume $\rho^{\otimes n}$ is a state on $\mathcal{H}^{\otimes n}$. Let $m\ge2\in \mathbb{N}$, $\delta\in (0,1)$. Then the conditional error exponent of the coherence distillation under $\mathbb{O}^{\delta}_{\mathcal{NG}}$ is 
\begin{align*}
	E_{d,err,p}^{(m),\mathbb{O}^{\delta}_{\mathcal{NG}}}(\rho^{\otimes n} )\le\frac{1}{n}\hat{\beta}_{\frac{\delta+1}{m},ic}(\rho^{\otimes n}).
\end{align*}

Besides, the zero-rate conditional error exponent of probabilistic coherence distillation under $\mathbb{O}_{\mathcal{NG}}^{\delta}$ for $\rho$ is
\begin{align*}
	E_{d,err,p}^{\infty,\mathbb{O}^{\delta}_{\mathcal{NG}}}(\rho)\le\hat{D}_{\Omega,ic}^{reg }(\rho).
\end{align*}
\begin{proof}
Assume $\{\mathcal{E}_i\}_{i=1}^k$ is a feasible instrument such that
\begin{align}
	1-\epsilon_n\le& F(\frac{\mathcal{E}_i(\rho^{\otimes n})}{\mathrm{tr}(\mathcal{E}_i(\rho^{\otimes n}))},\Psi_m)\nonumber\\
	=&\bra{\psi_m}\frac{\mathcal{E}_i(\rho^{\otimes n})}{\mathrm{tr}(\mathcal{E}_i(\rho^{\otimes n}))}\ket{\psi_m}\nonumber\\
	=&\frac{\mathrm{tr}[\rho_{AB}^{\otimes n}\mathcal{E}_i^{\dagger}(\Psi_m)]}{p_n},\label{cf0}
\end{align}
here $\mathcal{E}_i^{\dagger}(\cdot)$ satisfies $\mathrm{tr}(\mathcal{E}_i^{\dagger}(A)B)=\mathrm{tr}(A\mathcal{E}_i(B))$, $p_n=\mathrm{tr}(\mathcal{E}_i^{\dagger}(\mathbb{I})\rho_{AB}^{\otimes n})$.

Let $M^{(n)}_2=\mathcal{E}_i^{\dagger}(\Psi_m), M^{(n)}_1=\mathcal{E}^{\dagger}_i(\mathbb{I}-\Psi_m)$, $M^{(n)}_0=\mathbb{I}-M^{(n)}_1-M_2^{(n)}$. As $\mathcal{E}_i$ is completely positive trace nonincreasing and $\mathbb{I}-\Psi_m\ge 0$, $M^{(n)}_1,M^{(n)}_2\ge 0$. As $\sum_{i=1}^k\mathcal{E}_i$ is trace preserving, then $\sum_i\mathcal{E}_i^{\dagger}(\mathbb{I})=\mathbb{I}$, 
\begin{align*}
	M^{(n)}_0=&\mathbb{I}-M^{(n)}_1-M^{(n)}_2\\
	=&\sum_{\{1,2,\cdots,k\}-i}\mathcal{E}_l^{\dagger}(\mathbb{I})\ge 0,
\end{align*} 
Hence, $\{M^{(n)}_0,M^{(n)}_1,M^{(n)}_2\}$ is a POVM. Next based on (\ref{cf0}), we have
\begin{align}
	\frac{\mathrm{tr}M_1^{(n)}\rho^{\otimes n}}{\mathrm{ tr}(M_1^{(n)}+M_2^{(n)})\rho^{\otimes n}}\le \epsilon_n.
\end{align}

Assume $\sigma_n$ is an arbitrary incoherent state on $\mathcal{H}^{\otimes n}$, then
\begin{align*}
	\frac{\mathrm{tr}M^{(n)}_2\sigma_n}{\mathrm{tr}(M^{(n)}_1+M^{(n)}_2)\sigma_n}=&\frac{\mathrm{tr}(\mathcal{E}_i(\sigma_n)\Psi_{m_n})}{\mathrm{tr}\mathcal{E}_i(\sigma_n)}\\
	\le& \frac{\delta+1}{m_n},
\end{align*}
the last inequality is due to Lemma \ref{cm}.  Thus, based on the definition of postselected hypothesis testing, we have
\begin{align*}
	\hat{\beta}_{\frac{\delta+1}{m_n}, ic}(\rho^{\otimes n})\ge& \log\frac{\mathrm{tr}(M^{(n)}_1+M^{(n)}_2)\rho^{\otimes n}}{\mathrm{tr}(M^{(n)}_1\rho^{\otimes n})}\\
	\ge&-\log\epsilon_n.
\end{align*}
Then multiplying two sides $\frac{1}{n}$, and when taking the supermum over all $\{\mathcal{E}_i\}_{i=1}^k\in \mathbb{O}_{\mathcal{NG}}^{\delta},$ we have
\begin{align}
	E^{(m_n)}_{d,err,p}(\rho )\le&	\frac{1}{n}\min_{\sigma_n\in ic} \hat{\beta}_{\frac{\delta+1}{m_n}}(\rho^{\otimes n},\sigma_n)\nonumber\\
	=&\frac{1}{n} \hat{\beta}_{\frac{\delta+1}{m},ic}(\rho^{\otimes n}).
\end{align}

Let $n\rightarrow \infty$,
\begin{align}
	E_{d,err,p}^{\infty,\mathbb{O}^{\delta}_{\mathcal{NG}}}(\rho )\le \lim\limits_{n\rightarrow\infty}\frac{1}{n}\hat{\beta}_{\frac{\delta+1}{m},ic}(\rho^{\otimes n})=D_{\Omega,ic}^{reg}(\rho).\label{c3}
\end{align}

The last equality is due to Corollary \ref{c1}.
\end{proof}
\subsection{Probabilistic entanglement distillation exponents for the resource theory of magic}\label{ptm}

Assume $\mathcal{H}$ is a Hilbert space with $\dim(\mathcal{H})=3,$ and the standard computational basis is $\{\ket{0},\ket{1},\ket{2}\}$. Let 
\begin{align*}
	X=\sum_{j=0}^2\ket{j}\bra{j\oplus1},\hspace{5mm} Z=\sum_{j=0}^2 \omega^j\ket{j}\bra{j},
\end{align*} 
here $\oplus$ denotes addition modulo $d$ and $\omega=e^{\frac{2\pi i}{d}}.$ We define the Heisenberg–Weyl operators as
\begin{align*}
	T_{\boldsymbol{u}}=\tau^{-a_1a_2}Z^{a_1}X^{a_2},
\end{align*}
where $\tau=e^{\frac{4\pi i}{3}}$ and $\boldsymbol{u}=(a_1,a_2)\in \mathbb{Z}_3\times \mathbb{Z}_3.$

The set of Clifford operators and stabilizer states are defined as follows, respectively,
\begin{align*}
	\mathcal{C}_3=&\{U|UT_{\boldsymbol{u}}U^{\dagger}=e^{i\theta}T_{\boldsymbol{v}}, \forall \boldsymbol{u}, \exists \theta,\boldsymbol{v}\},\\
		STAB(\mathcal{H}_3)=&\{\rho\in \mathcal{D}(\mathcal{H}_3)|\rho=\sum_i p_i U_i\ket{0}\bra{0}U_i^{\dagger}, U\in \mathcal{C}_3\},
\end{align*}
where $\{p_i\}_i$ is a set of probabilistic distribution with $p_i\ge 0$ and $\sum_ip_i=1.$

The generalized robustness of magic for a states $\rho$ is defined as follows \cite{howard2017application},
\begin{align*}
	R_{\boldsymbol{G}}(\rho)=\min\{s\ge 0|\frac{\rho+s\tau}{1+s}\in STAB,\tau\in \mathcal{D}(\mathcal{H})\}.
\end{align*}

Next we define the following twirling operation as follows,
\begin{align*}
	\mathcal{T}_{Sp}(\cdot)=\frac{1}{|Sp(2,3)|}\sum_{F\in Sp(2,3)}U_F(\cdot) U_F^{\dagger},
\end{align*}
where $Sp(2,3)=\{F\in GL(2,\mathbb{F}_3)|F^T\begin{pmatrix}
	0&1\\
	-1&0
\end{pmatrix}F=\begin{pmatrix}
0&1\\-1&0
\end{pmatrix}\},$ $GL(2,\mathbb{F}_3)=\{\begin{pmatrix}
a&b\\
c&d
\end{pmatrix}|ad-bc\ne 0, a,b,c,d\in \mathbb{F}_3\}.$ In \cite{Veitch_2014}, the authors showed that $\ket{S}$ stays invariant under the twirling operation $\mathcal{T}_{Sp}(\cdot)$, and for any state $\rho$,
\begin{align*}
	\mathcal{T}_{Sp}(\rho)=F_S\ket{S}\bra{S}+\frac{1-F_S}{2}(\mathbb{I}-\ket{S}\bra{S}),
\end{align*}
where $F_S=\bra{S}\rho\ket{S}.$

\begin{lemma}\label{mrs}
	Assume $\rho_F=F\mathbb{S}+\frac{1-F}{2}(\mathbb{I}-\mathbb{S})$, here $\mathbb{S}=\ket{S}\bra{S}$, then
	\begin{align*}
		R_{\boldsymbol{G}}(\rho_F)=& \max(0,2F-1).\\
		\max_{\sigma\in STAB}&\mathrm{ tr} \sigma \mathbb{S}= \frac{1}{2}.
	\end{align*}
\end{lemma}

The above results were shown in \cite{Veitch_2014,one2022takagi}.

\begin{lemma}\label{mm}
	Assume $\mathcal{H}$ is a Hilbert space, both $M$ and $N$ are semidefinite positive operators with $M+N\le \mathbb{I}$, let $$\Lambda_{M,N}(X)=\mathrm{ tr}MX\cdot \mathbb{S}+\mathrm{ tr}NX\cdot \varphi,$$ here $\varphi=\frac{\mathbb{I}-\mathbb{S}}{2}$, then for any $\delta>0$, then
	$\Lambda_{M,N}(\cdot)\in \mathcal{NG}_{\delta}$ if and only if $\sup_{X\in STAB}\frac{\mathrm{ tr}MX}{\mathrm{ tr}NX}\le \frac{1+\delta}{1-\delta}$.
\end{lemma}
\begin{proof}
	Based on Lemma \ref{mrs}, for any $\sigma\in \mathcal{D}(\mathcal{H}),$
	\begin{align*}
		R_{g,m}(\frac{\Lambda_{M,N}(\sigma)}{\mathrm{ tr}\Lambda_{M,N}(\sigma)})=\max(0,\frac{2\mathrm{ tr}M\sigma}{\mathrm{ tr}(M+N)\sigma}-1),
	\end{align*}
	Hence, $\Lambda_{M,N}(\cdot)\in \mathcal{NG}_{\delta}\Longleftrightarrow \sup_{X\in STAB}\frac{\mathrm{ tr}MX}{\mathrm{ tr}NX}\le \frac{1+\delta}{1-\delta}.$
\end{proof}

\emph{Theorem \ref{edm}:} 	Assume $\rho$ is a state on $\mathcal{H}$, let $\delta\in(0,1)$, then the conditional error exponent of the probabilistic magic distillation under $\mathbb{O}^{\delta}_{\mathcal{NG}}$ is 
\begin{align*}
	E_{d,err,p}^{\infty,\mathbb{O}^{\delta}_{\mathcal{NG}}}(\rho)=\hat{D}_{\Omega,STAB}^{reg }(\rho),
\end{align*}
here $$\mathbb{O}^{\delta}_{\mathcal{NG}}=\{\mathcal{E}=\{\mathcal{E}_i\}|R_{\boldsymbol{G}}(\frac{\mathcal{E}_i(\sigma)}{\mathrm{ tr}\mathcal{E}_i(\sigma)})\le \delta,\forall\sigma\in STAB\}.$$

\begin{proof}
	Based on Lemma \ref{lf},
	 	\begin{align*}
	 	E_{d,err,p}^{\mathbb{O}^{\delta}_{\mathcal{NG}}}(\rho^{\otimes n})=&\sup-\frac{1}{n}\log\epsilon_n\\
	 	\textit{s. t.}\hspace{4mm}&  \frac{\mathrm{tr}M^{(n)}_2\rho^{\otimes n}}{\mathrm{ tr}(M^{(n)}_1+M^{(n)}_2)\rho^{\otimes n}}\ge 1-\epsilon_n,
	 	\mathcal{E}_i(X)=\mathrm{ tr}M^{(n)}_2X\cdot \mathbb{S}+\mathrm{ tr}M^{(n)}_1X\cdot \varphi,\\
	 	&M^{(n)}_2+M^{(n)}_1\le \mathbb{I},M^{(n)}_1,M^{(n)}_2\ge 0,\\
	 	&\mathcal{E}_i\in \mathcal{E},\mathcal{E}\in\mathbb{O}^{\delta}_{\mathcal{NG}}.
	 \end{align*}
	 
	 Then let $\{\mathcal{E}_i\}_{i=1}^k$ is a feasible $\mathbb{O}_{\mathcal{NG}}^{\delta}$ instrument such that
	 \begin{align}
	 	1-\epsilon_n\le& F(\frac{\mathcal{E}_i(\rho^{\otimes n})}{\mathrm{tr}(\mathcal{E}_i(\rho^{\otimes n}))},\mathbb{S})\nonumber\\
	 	=&\bra{\mathbb{S}}\frac{\mathcal{E}_i(\rho^{\otimes n})}{\mathrm{tr}(\mathcal{E}_i(\rho^{\otimes n}))}\ket{\mathbb{S}}\nonumber\\
	 	=&\frac{\mathrm{tr}[\rho_{AB}^{\otimes n}\mathcal{E}_i^{\dagger}(\mathbb{S})]}{p_n},\label{tfm0}
	 \end{align}
	 here $\mathcal{E}_i^{\dagger}(\cdot)$ satisfies $\mathrm{tr}(\mathcal{E}_i^{\dagger}(A)B)=\mathrm{tr}(A\mathcal{E}_i(B))$, $p_n=\mathrm{tr}(\mathcal{E}_i^{\dagger}(\mathbb{I})\rho_{AB}^{\otimes n})$.
	 
	 Let $M^{(n)}_2=\mathcal{E}_i^{\dagger}(\mathbb{S}), M^{(n)}_1=\mathcal{E}^{\dagger}_i(\mathbb{I}-\mathbb{S})$, $M^{(n)}_0=\mathbb{I}-M^{(n)}_1-M_2^{(n)}$. As $\mathcal{E}_i$ is completely positive trace nonincreasing and $\mathbb{I}-\mathbb{S}\ge 0$, $M^{(n)}_1,M^{(n)}_2\ge 0$. Let
	 \begin{align*}
	 	M^{(n)}_0=&\mathbb{I}-M^{(n)}_1-M^{(n)}_2\\
	 	=&\sum_{l\in\{1,2,\cdots,k\}-i}\mathcal{E}_l^{\dagger}(\mathbb{I})\ge 0,
	 \end{align*} 
	 next because $\sum_{i=1}^k\mathcal{E}_i$ is trace preserving, $\{M^{(n)}_0,M^{(n)}_1,M^{(n)}_2\}$ is a POVM. Based on (\ref{tfm0}), we have
	 \begin{align}
	 	\frac{\mathrm{tr}M_1^{(n)}\rho^{\otimes n}}{\mathrm{ tr}(M_1^{(n)}+M_2^{(n)})\rho^{\otimes n}}\le \epsilon_n.
	 \end{align}
	 
	 Assume $\sigma_n$ is an arbitrary stabilizer state on $\mathcal{H}^{\otimes n}$, then
	 \begin{align*}
	 	\frac{\mathrm{tr}M^{(n)}_2\sigma_n}{\mathrm{tr}(M^{(n)}_1+M^{(n)}_2)\sigma_n}=&\frac{\mathrm{tr}(\mathcal{E}_i(\sigma_n)\mathbb{S})}{\mathrm{tr}\mathcal{E}_i(\sigma_n)}\\
	 	\le& \frac{\delta+1}{2},
	 \end{align*}
	 the last inequality is due to Lemma \ref{mm}.  Thus, based on the definition of postselected hypothesis testing, we have
	 \begin{align*}
	 	\hat{\beta}_{\frac{\delta+1}{2}, STAB}(\rho^{\otimes n})\ge& \log\frac{\mathrm{tr}(M^{(n)}_1+M^{(n)}_2)\rho^{\otimes n}}{\mathrm{tr}(M^{(n)}_1\rho^{\otimes n})}\\
	 	\ge&-\log\epsilon_n.
	 \end{align*}
	 Then multiplying two sides $\frac{1}{n}$, and when taking the supermum over all $\{\mathcal{E}_i\}_{i=1}^k\in \mathbb{O}_{\mathcal{NG}}^{\delta},$ we have
	 \begin{align}
	 	E_{d,err,p}^{\mathbb{O}^{\delta}_{\mathcal{NG}}}(\rho^{\otimes n})\le&	\frac{1}{n}\min_{\sigma_n\in STAB} \hat{\beta}_{\frac{\delta+1}{2}}(\rho^{\otimes n},\sigma_n)\nonumber\\
	 	=&\frac{1}{n} \hat{\beta}_{\frac{\delta+1}{2},STAB}(\rho^{\otimes n}).\label{am1}
	 \end{align}
	 Next we show the other direction. Assume $\{(M^{(n)}_1,M^{(n)}_2,M^{(n)}_0)|\sum_{i=0}^2M_i^{(n)}=\mathbb{I},M^{(n)}_i\ge0,i=1,2,0\}$ is a POVM with $\frac{\mathrm{tr}M^{(n)}_2\sigma_n}{\mathrm{ tr}(M_1^{(n)}\sigma_n+M^{(n)}_2\sigma_n)}\le \frac{\delta+1}{2}$ for any $\sigma_n\in STAB.$ Let $\mathcal{E}_1(\cdot)$, $\mathcal{E}_2(\cdot)$ be
	 \begin{align*}
	 	\mathcal{E}_1(\cdot)=\mathrm{tr}(M^{(n)}_2(\cdot))\mathbb{S}+\mathrm{tr}M^{(n)}_1(\cdot)\frac{\mathbb{I}-\mathbb{S}}{2},\\
	 	\mathcal{E}_2(\cdot)=\mathrm{tr}M_0^{(n)}(\cdot)\frac{\mathbb{I}-\mathbb{S}}{2}.
	 \end{align*}
	 As $\frac{\mathbb{I}-\mathbb{S}}{2}\in STAB$, $\mathcal{E}_2$ is $\mathcal{NG}$. Based on Lemma \ref{dm}, and  $\frac{\mathrm{tr}M^{(n)}_2\sigma_n}{\mathrm{ tr}(M_1^{(n)}\sigma_n+M^{(n)}_2\sigma_n)}\le \frac{\delta+1}{2}$, $\mathcal{E}_1(\cdot)$ is in $\mathbb{O}_{\mathcal{NG}}^{\delta}.$Then 
	 \begin{align*}
	 	E_{d,err,p}^{\mathbb{O}^{\delta}_{\mathcal{NG}}}(\rho^{\otimes n})\ge& -\frac{1}{n}\log\frac{\mathrm{tr}M_1^{(n)}\rho^{\otimes n}}{\mathrm{tr}(M_1^{(n)}+M_2^{(n)})\rho^{\otimes n}}
	 \end{align*}
	 by optimising over all measurements with the property, we have
	 \begin{align}
	 	E_{d,err,p}^{\mathbb{O}^{\delta}_{\mathcal{NG}}}(\rho^{\otimes n})\ge& \frac{1}{n}\hat{\beta}_{\frac{\delta+1}{2},STAB}(\rho^{\otimes n}).\label{am2}
	 \end{align}
	 
	 Based on (\ref{am1}) and (\ref{am2}), we have
	 \begin{align}
	 	E^{\infty}_{d,err,p}(\rho)= \lim\limits_{n\rightarrow\infty}\frac{1}{n}\hat{\beta}_{\frac{\delta+1}{2},STAB}(\rho^{\otimes n})\label{am3}
	 \end{align}

	 Combing Corollary \ref{c1} and (\ref{am3}), we have
	 \begin{align*}
	 	E^{\infty}_{d,err,p}(\rho)=&D_{\Omega,STAB}^{reg}(\rho).
	 \end{align*}
\end{proof}

\emph{Example \ref{e3}} 	Let 
\begin{align*}
	\chi_p=&p\ket{\phi}\bra{\phi}+(1-p)\frac{\mathbb{I}}{3}, \\\ket{\phi}=&\frac{1}{\sqrt{6}}(2\ket{0}-\ket{1}-\ket{2}).
\end{align*}
Then for each $n\in \mathbb{N}$, 
\begin{align*}
	\frac{1}{n}D_{\Omega,STAB}(\chi_p^{\otimes n})=D_{\Omega,STAB}(\chi_p)=\begin{cases}
		0&\hspace{5mm} p\in [0,\frac{2}{5}],\\
		\log\frac{1+2p}{3(1-p)}&\hspace{5mm}p\in (\frac{2}{5},1).
	\end{cases}
\end{align*}

\begin{proof}
	As when $p=\frac{2}{5}$, 
	\begin{align*}
		&\rho_{\frac{2}{5}}=\frac{1}{5}(\ket{0}\bra{0}+\ket{u_1}\bra{u_1}+\ket{u_2}\bra{u_2}+\ket{u_3}\bra{u_3}+\ket{u_4}\bra{u_4}),\\
	&\ket{u_1}=\frac{1}{\sqrt{3}}\begin{pmatrix}
			1\\\omega\\\omega^2
		\end{pmatrix},\hspace{3mm}\ket{u_2}=\frac{1}{\sqrt{3}}\begin{pmatrix}
		1\\\omega^2\\\omega
		\end{pmatrix},\hspace{3mm}\ket{u_3}=\frac{1}{\sqrt{3}}\begin{pmatrix}
		1\\\omega\\\omega
		\end{pmatrix},\hspace{3mm}\ket{u_4}=\frac{1}{\sqrt{3}}\begin{pmatrix}
		1\\\omega^2\\\omega^2
		\end{pmatrix}.
	\end{align*}
	Hence, when $p\le \frac{2}{5}$, $\rho_p\in STAB.$ Furthermore, when $p>\frac{2}{5}$, there exists negative value in the Wigner function of $\rho_p$. Then when $p\in (\frac{2}{5},1)$,
	\begin{align}
		\Omega_{ STAB}(\rho_p)\le \Omega(\rho_p,\rho_{\frac{2}{5}})=\log\frac{1+2p}{3(1-p)},\label{mf1}
	\end{align}
	
	Next we show the other direction. 	The dual problem of $\Omega_{STAB}(\cdot)$ is
	\begin{align}
		\Omega_{STAB}(\rho)=\hspace{2mm}&\hspace{2mm}\sup\mathrm{tr}(A\rho)\label{dmcc}\\
		\textit{s. t.}\hspace{4mm}&\mathrm{tr}B\rho=1,\nonumber\\
		&\mathrm{tr}(B-A)\sigma\ge 0\hspace{3mm}\forall \sigma\in STAB,\nonumber\\
		&A,B\ge 0.\nonumber
	\end{align} 
	Let
	\begin{align*}
		A=\frac{\Phi}{1-p},\hspace{3mm}B=\frac{2\ket{u}\bra{u}+\ket{v}\bra{v}}{1-p},
	\end{align*}
	here $$\ket{u}=\frac{1}{\sqrt{3}}\begin{pmatrix}
		1\\1\\1
	\end{pmatrix},\ket{v}=\frac{1}{\sqrt{2}}\begin{pmatrix}
	0\\1\\-1
	\end{pmatrix}.$$
	Due to computation, $$\mathrm{tr}A\rho=\frac{1+2p}{3(1-p)},\mathrm{tr}B\rho=1.$$ As $STAB$ is constructed by pure qutrit stabilizer states, then for any pure qutrit stabilizer state $\gamma$, $$\mathrm{ tr}(B-A)\gamma\ge 0 \Longrightarrow\forall\mathrm{ tr}(B-A)\sigma\ge 0,\qquad \forall\sigma\in STAB.$$
	Hence, $A$ and $B$ are feasible for the dual program of $\Omega_{Sep}(\cdot)$, then
	\begin{align}
		D_{\Omega,{STAB}}(\rho)\ge \log\frac{1+2p}{3(1-p)}. \label{mf2}
	\end{align}
	Combing (\ref{mf1}) and (\ref{mf2}), when $p\in (\frac{2}{5},1),$
	\begin{align}
		D_{\Omega,STAB}(\rho_p)=\log\frac{1+2p}{3(1-p)}.\label{lmf2}
	\end{align}
	For the state $\rho_p^{\otimes n}$, 
	based on the property of $\Omega(\cdot,\cdot)$, 
	\begin{align}
		\Omega_{STAB_n}(\rho^{\otimes n})\ge \Omega(\rho_p^{\otimes n},\rho_{\frac{2}{5}}^{\otimes n})=n\Omega(\rho_p,\rho_{\frac{2}{5}}). \label{lmff2}
	\end{align}
	Next let
	\begin{align*}
		A^{(n)}=A^{\otimes n},B^{(n)}=B^{\otimes n},
	\end{align*}
	Due to the computation, $\mathrm{tr}B^{(n)}\rho_p^{\otimes n}=\mathrm{ tr}^nB\rho_p=1,$ and for any $\sigma_n\in STAB(\mathcal{H}^{\otimes n})$,
	\begin{align*}
		&\mathrm{tr}(B^{(n)}-A^{(n)})\sigma_n\\
		=&3^n\sum_{\boldsymbol{u}}W_{B^{(n)}-A^{(n)}}(\boldsymbol{u})W_{\sigma_n}(\boldsymbol{u})\ge 0,
	\end{align*}
	here $W_O(\boldsymbol{u})=\frac{1}{3}\mathrm{ tr}T_{\boldsymbol{u}}O$.
	Hence, $B^{(n)}-A^{(n)}\in (cone(STAB_n))^{*},$ that is, $A^{(n)}$ and $B^{(n)}$ are the feasible for $\rho_p^{\otimes n}$ in terms of $(\ref{dmcc})$ for $\Omega_{Sep}$, 
	\begin{align}
		\frac{1}{n}	D_{\Omega,Sep}(\rho^{\otimes n})\ge\log\frac{1+2p}{3(1-p)}.\label{rmf2}
	\end{align}
	Combing (\ref{lmff2}) and (\ref{rmf2}), we have $$\frac{1}{n}D_{\Omega,Sep}(\rho_p)=\log\frac{1+2p}{3(1-p)}.$$
\end{proof}

\begin{lemma}\label{mf6}
	Assume $\mathcal{H}$ is Hilbert space with $\dim(\mathcal{H})=3$. Let $(STAB_n)_n$ be a sequence of sets of states, $r_n\in (0,1]$, and 
	\begin{align*}
		\mathbb{M}_n=\left\{\left(\frac{\mathbb{I}^{\otimes n}+X_n}{2},\frac{\mathbb{I}^{\otimes n}-X_n}{2}\right)|X_n=X_n^{\dagger}\in B(\mathcal{H}^{\otimes n}), ||X_n||_{\infty}\le r_n\right\},
	\end{align*}
	where $0<r_n\leq \frac{1}{2\cdot 3^n+1},$ $||Y||_{\infty}$ is the operator norm and $B(\mathcal{H})$ is the set of bounded operators. Then for any $n,m\in \mathbb{N}^{+}$, for any state $\sigma_{n+m}\in STAB_{n+m}$, it holds that 
	\begin{align*}
		\mathrm{ tr}_{1\cdots n}[(\frac{\mathbb{I}^{\otimes n}\pm X_n}{2}\otimes \mathbb{I}^{\otimes m})\sigma_{n+m}]\in cone(STAB_{m}),
	\end{align*}
	where $cone(STAB_{m})=\{\lambda\gamma|\lambda\ge 0,\gamma\in STAB_{m}\}.$	
\end{lemma}

\begin{proof}
By convexity, it suffices to consider a pure stabilizer state
\(\sigma_{n+m}=|\phi\rangle\langle\phi|\). By the bipartite normal form
of stabilizer states, there exists local Clifford unitaries
\(U_A,U_B\), an integer $
0\le k\le \min\{n,m\},
$ and decompositions \(A=A_1A_2\), \(B=B_1B_2\), where \(A_1\) and
\(B_1\) each contain \(k\) qutrits, such that
$$
(U_A\otimes U_B)|\phi\rangle
=
|\Phi_3\rangle_{A_1B_1}^{\otimes k}
\otimes |0\rangle_{A_2}
\otimes |0\rangle_{B_2},
$$
where
$$
|\Phi_3\rangle
=
\frac{1}{\sqrt3}
\sum_{j=0}^{2}|jj\rangle .
$$
Set
$$
X'_n=U_AX_nU_A^\dagger,
\qquad
Z=\langle0|_{A_2}X'_n|0\rangle_{A_2}.
$$
Then \(Z=Z^\dagger\) and
$$
\|Z\|_\infty\le\|X_n\|_\infty\le r_n.
$$

Writing \(d=3^k\), the maximally entangled-state identity gives

\begin{align}
	&\operatorname{Tr}_{A}
	\left[
	\left(
	\frac{\mathbb I^{\otimes n}\pm X_n}{2}
	\otimes\mathbb I^{\otimes m}
	\right)
	|\phi\rangle\langle\phi|
	\right]
=
	U_B^\dagger
	\left[
	\frac{\mathbb I_d\pm Z^{T}}{2d}
	\otimes
	|0\rangle\langle0|_{B_2}
	\right]
	U_B .
\end{align}

Let \(\{P_s\}_{s=1}^{N_k}\) be the set of pure \(k\)-qutrit stabilizer
projectors. Since these projectors form a complex projective
\(2\)-design, for any Hermitian operator \(Y\),
\begin{align*}
	\sum_{s=1}^{N_k}P_s=&\frac{N_k}{d}\mathbb I_d,\\
	\sum_{s=1}^{N_k}
	\operatorname{Tr}(P_sY)P_s
	=&
	\frac{N_k}{d(d+1)}
	\left(Y+\operatorname{Tr}(Y)\mathbb I_d\right).
\end{align*}
When taking \(Y=Z^T\),
$$
\frac{\mathbb I_d\pm Y}{2d}
=
\frac{1}{2N_k}
\sum_{s=1}^{N_k}
c_s^{\pm}P_s,
$$
here
$$
	c_s^{\pm}
=
1
\pm(d+1)\operatorname{Tr}(P_sY)
\mp\operatorname{Tr}Y
$$

As 
\begin{align*}
	c_s^{\pm}
	=&
	1
	\pm(d+1)\operatorname{Tr}(P_sY)
	\mp\operatorname{Tr}Y\\
	\ge &1-(d+1)r_n-dr_n\\
	=&	1-(2d+1)r_n\\
	\ge&0.
\end{align*}
here the first inequality is due to that $|\operatorname{Tr}(P_sY)|
\le\|Y\|_\infty
\le r_n,$ and $|\operatorname{Tr}Y|
\le d\|Y\|_\infty
\le dr_n,$
the second inequality is due to that \(d=3^k\le3^n\). Therefore,
$$
c_s^\pm
\ge
1-(2\cdot3^n+1)r_n
\ge0.
$$
Consequently,

$$
\frac{\mathbb I_d\pm Z^T}{2d}
\in
\operatorname{cone}(\operatorname{STAB}_k).
$$

Tensoring with the stabilizer state
\(|0\rangle\langle0|_{B_2}\) and conjugating by the Clifford unitary
\(U_B\) preserve the stabilizer cone. Thus

$$
\operatorname{Tr}_{A}
\left[
\left(
\frac{\mathbb I^{\otimes n}\pm X_n}{2}
\otimes\mathbb I^{\otimes m}
\right)
|\phi\rangle\langle\phi|
\right]
\in
\operatorname{cone}(\operatorname{STAB}_m).
$$

The claim for mixed stabilizer states follows by convexity.
\end{proof}

\begin{example}
	Let
	\begin{equation}
		\ket{\Psi}
		=
		\frac{1}{\sqrt{6}}
		\left(2\ket{0}-\ket{1}-\ket{2}\right)
	\end{equation}
	be the qutrit Norrell state, and define
	\begin{equation}
		\rho_p
		=
		p\Psi+(1-p)\frac{\mathbb{I}}{3},
		\qquad
		\Psi=\ket{\Psi}\!\bra{\Psi},
		\qquad
		0\leq p\leq 1.
	\end{equation}
	For $2/5<p<1$, let $y_p>1$ be the unique solution of
	\begin{equation}
		y_p^2(2y_p+1)
		=
		\frac{1+2p}{1-p}.
		\label{eq:y-equation}
	\end{equation}
	Then the unique minimizer of
	\begin{equation}
		D(\operatorname{STAB}_1\Vert\rho_p)
		:=
		\min_{\sigma\in\operatorname{STAB}_1}
		D(\sigma\Vert\rho_p)
	\end{equation}
	is
	\begin{align}
		\sigma_p^\star
		={}&
		\frac{y_p}{3y_p+2}
		\left(
		\ket{0}\!\bra{0}
		+
		\ket{1,0}\!\bra{1,0}
		+
		\ket{2,0}\!\bra{2,0}
		\right)
		\nonumber\\
		&+
		\frac{1}{3y_p+2}
		\left(
		\ket{0,1}\!\bra{0,1}
		+
		\ket{0,2}\!\bra{0,2}
		\right),
		\label{eq:optimal-stabilizer-state}
	\end{align}
	where
	\begin{equation}
		\ket{m,k}
		=
		\frac{1}{\sqrt{3}}
		\sum_{x=0}^{2}
		\omega^{mx^2+kx}\ket{x},
		\qquad
		\omega=e^{2\pi i/3}.
	\end{equation}
	Consequently,
	\begin{equation}
		D(\operatorname{STAB}_1\Vert\rho_p)
		=
		D(\sigma_p^\star\Vert\rho_p)=	\begin{cases}
			0,
			&
			0\leq p\leq \dfrac{2}{5},
			\\[3mm]
			\displaystyle
			q_\Psi\log\frac{q_\Psi}{a_p}
			+
			q_u\log\frac{q_u}{b_p}
			+
			q_v\log\frac{q_v}{b_p},
			&
			\dfrac{2}{5}<p<1.
		\end{cases}
	\end{equation}
	where $q_\Psi=\frac{2y_p+1}{3y_p+2},q_u=\frac{y_p}{3y_p+2},q_v=\frac{1}{3y_p+2},a_p=\frac{1+2p}{3},$
	and $b_p=\frac{1-p}{3}.$
\end{example}
\begin{proof}
	Introduce the orthonormal basis
	\begin{equation}
		\ket{\Psi}
		=
		\frac{1}{\sqrt{6}}(2,-1,-1)^{\mathsf T},
		\qquad
		\ket{u}
		=
		\frac{1}{\sqrt{3}}(1,1,1)^{\mathsf T},
		\qquad
		\ket{v}
		=
		\frac{1}{\sqrt{2}}(0,1,-1)^{\mathsf T},
	\end{equation}
	and write
	\begin{equation}
		U=\ket{u}\!\bra{u},
		\qquad
		V=\ket{v}\!\bra{v}.
	\end{equation}
	Since $\mathbb{I}=\Psi+U+V$, the state $\rho_p$ has the
	spectral decomposition
	\begin{equation}
		\rho_p
		=
		a_p\Psi+b_pU+b_pV,
		\qquad
		a_p=\frac{1+2p}{3},
		\qquad
		b_p=\frac{1-p}{3}.
		\label{eq:rho-spectral}
	\end{equation}
	
	A direct calculation gives
	\begin{align}
		\ket{0}\!\bra{0}
		+
		\ket{1,0}\!\bra{1,0}
		+
		\ket{2,0}\!\bra{2,0}
		&=
		2\Psi+U,
		\label{eq:first-projector-identity}\\
		\ket{0,1}\!\bra{0,1}
		+
		\ket{0,2}\!\bra{0,2}
		&=
		\Psi+V.
		\label{eq:second-projector-identity}
	\end{align}
	Hence, setting $y=y_p$, Eq.~\eqref{eq:optimal-stabilizer-state}
	can be written as
	\begin{equation}
		\sigma_p^\star
		=
		q_\Psi\Psi+q_uU+q_vV,
		\label{eq:sigma-spectral}
	\end{equation}
	where
	\begin{equation}
		q_\Psi=\frac{2y+1}{3y+2},
		\qquad
		q_u=\frac{y}{3y+2},
		\qquad
		q_v=\frac{1}{3y+2}.
		\label{eq:q-definitions}
	\end{equation}
	The coefficients in Eq.~\eqref{eq:optimal-stabilizer-state}
	are nonnegative and sum to one:
	\begin{equation}
		3\frac{y}{3y+2}
		+
		2\frac{1}{3y+2}
		=1.
	\end{equation}
	Therefore,
	\begin{equation}
		\sigma_p^\star\in\operatorname{STAB}_1.
	\end{equation}
	
	We now introduce the Hermitian operator
	\begin{equation}
		W=-\Psi+2U+V
		=
		\begin{pmatrix}
			0&1&1\\
			1&1&0\\
			1&0&1
		\end{pmatrix}.
		\label{eq:stabilizer-witness}
	\end{equation}
	Its expectation values on the twelve pure qutrit stabilizer
	states are
	\begin{equation}
		\bra{s}W\ket{s}
		=
		\begin{cases}
			0,
			&
			\ket{s}\in
			\{
			\ket{0},
			\ket{1,0},
			\ket{2,0},
			\ket{0,1},
			\ket{0,2}
			\},
			\\[2mm]
			1,
			&
			\ket{s}\in
			\{
			\ket{1},
			\ket{2},
			\ket{1,1},
			\ket{1,2},
			\ket{2,1},
			\ket{2,2}
			\},
			\\[2mm]
			2,
			&
			\ket{s}=\ket{0,0}.
		\end{cases}
		\label{eq:witness-expectations}
	\end{equation}
	Since $\operatorname{STAB}_1$ is the convex hull of these
	twelve projectors, Eq.~\eqref{eq:witness-expectations} implies
	\begin{equation}
		\operatorname{Tr}(W\sigma)\geq 0
		\qquad
		\text{for all }
		\sigma\in\operatorname{STAB}_1.
		\label{eq:witness-dual-cone}
	\end{equation}
	Moreover, $\sigma_p^\star$ is a convex combination of precisely
	the five stabilizer states on which $W$ vanishes. Thus,
	\begin{equation}
		\operatorname{Tr}(W\sigma_p^\star)=0.
		\label{eq:witness-complementarity}
	\end{equation}
	
	It remains to verify the first-order optimality condition.
	Because $\rho_p$ and $\sigma_p^\star$ are diagonal in the same
	basis, define
	\begin{equation}
		G
		:=
		\log\sigma_p^\star-\log\rho_p
		=
		g_\Psi\Psi+g_uU+g_vV,
	\end{equation}
	where
	\begin{equation}
		g_\Psi=\log\frac{q_\Psi}{a_p},
		\qquad
		g_u=\log\frac{q_u}{b_p},
		\qquad
		g_v=\log\frac{q_v}{b_p}.
	\end{equation}
	Equation~\eqref{eq:q-definitions} gives
	\begin{equation}
		g_u-g_v
		=
		\log\frac{q_u}{q_v}
		=
		\log y.
		\label{eq:first-gradient-difference}
	\end{equation}
	Furthermore, using Eq.~\eqref{eq:y-equation} and
	$q_\Psi/q_v=2y+1$, we obtain
	\begin{align}
		g_v-g_\Psi
		&=
		\log
		\left(
		\frac{q_v/b_p}{q_\Psi/a_p}
		\right)
		\nonumber\\
		&=
		\log
		\left(
		\frac{a_p}{b_p}
		\frac{q_v}{q_\Psi}
		\right)
		\nonumber\\
		&=
		\log
		\left(
		y^2(2y+1)\frac{1}{2y+1}
		\right)
		\nonumber\\
		&=
		2\log y.
		\label{eq:second-gradient-difference}
	\end{align}
	Since the eigenvalues of $W$ in the ordered basis
	$(\ket{\Psi},\ket{u},\ket{v})$ are $(-1,2,1)$,
	Eqs.~\eqref{eq:first-gradient-difference} and
	\eqref{eq:second-gradient-difference} imply that there exists
	$c_p\in\mathbb{R}$ such that
	\begin{equation}
		\log\sigma_p^\star-\log\rho_p	=	c_p\mathbb{I}+(\log y)W.	\label{eq:gradient-normal-form}
	\end{equation}
	
	For fixed full-rank $\rho_p$, the map
	\begin{equation}
		\sigma\longmapsto D(\sigma\Vert\rho_p)
	\end{equation}
	is convex. Its supporting-hyperplane inequality at
	$\sigma_p^\star$ gives, for every density operator $\sigma$,
	\begin{align}
		D(\sigma\Vert\rho_p)
		\geq{}&
		D(\sigma_p^\star\Vert\rho_p)
		\nonumber\\
		&+
		\operatorname{Tr}
		\left[
		\left(
		\log\sigma_p^\star
		+\mathbb{I}
		-\log\rho_p
		\right)
		(\sigma-\sigma_p^\star)
		\right].
		\label{eq:relative-entropy-support}
	\end{align}
	Since $\operatorname{Tr}(\sigma-\sigma_p^\star)=0$, the identity
	terms vanish. Using Eqs.~\eqref{eq:gradient-normal-form} and
	\eqref{eq:witness-complementarity}, we find
	\begin{align}
		D(\sigma\Vert\rho_p)
		&\geq
		D(\sigma_p^\star\Vert\rho_p)
		+
		(\log y)
		\operatorname{Tr}
		\left[
		W(\sigma-\sigma_p^\star)
		\right]
		\nonumber\\
		&=
		D(\sigma_p^\star\Vert\rho_p)
		+
		(\log y)\operatorname{Tr}(W\sigma).
		\label{eq:final-support-inequality}
	\end{align}
	For $p>2/5$, Eq.~\eqref{eq:y-equation} has $y>1$, and hence
	$\log y>0$. Therefore, Eq.~\eqref{eq:witness-dual-cone} implies
	that
	\begin{equation}
		D(\sigma\Vert\rho_p)
		\geq
		D(\sigma_p^\star\Vert\rho_p)
		\qquad
		\text{for all }
		\sigma\in\operatorname{STAB}_1.
	\end{equation}
	This proves the global optimality of $\sigma_p^\star$.
	
	Finally, since $D(\sigma\Vert\rho_p)$ is strictly convex in
	$\sigma$ for full-rank $\rho_p$, the minimizer is unique.
\end{proof}

	\end{document}